\documentclass[a4paper,USenglish,cleveref, autoref, thm-restate]{lipics-v2021}

\pdfoutput=1 %
\hideLIPIcs  %

\title{Scalable Distributed String Sorting} %

\author{Florian Kurpicz}{Karlsruhe Institute of Technology, Germany}{kurpicz@kit.edu}{https://orcid.org/0000-0002-2379-9455}{}
\author{Pascal Mehnert}{Independent, Germany}{pascalmehnert@posteo.de}{}{}
\author{Peter Sanders}{Karlsruhe Institute of Technology, Germany}{sanders@kit.edu}{https://orcid.org/0000-0003-3330-9349}{}
\author{Matthias Schimek}{Karlsruhe Institute of Technology, Germany}{schimek@kit.edu}{https://orcid.org/0009-0002-6402-9016}{}

\authorrunning{F. Kurpicz, P. Mehnert, P. Sanders, and M. Schimek} %

\Copyright{F. Kurpicz, P. Mehnert, P. Sanders, and M. Schimek} %



\category{} %

\relatedversion{This paper is based on Pascal Mehnert's Master's thesis~\cite{PascalStringSorting}. A brief announcement (3 pages) of the algorithms proposed in this paper will appear at SPAA'24.} %

\supplement{The source code of the implementations presented in this paper can be found at \url{https://github.com/pmehnert/distributed-string-sorting/}.}%

\funding{\flag{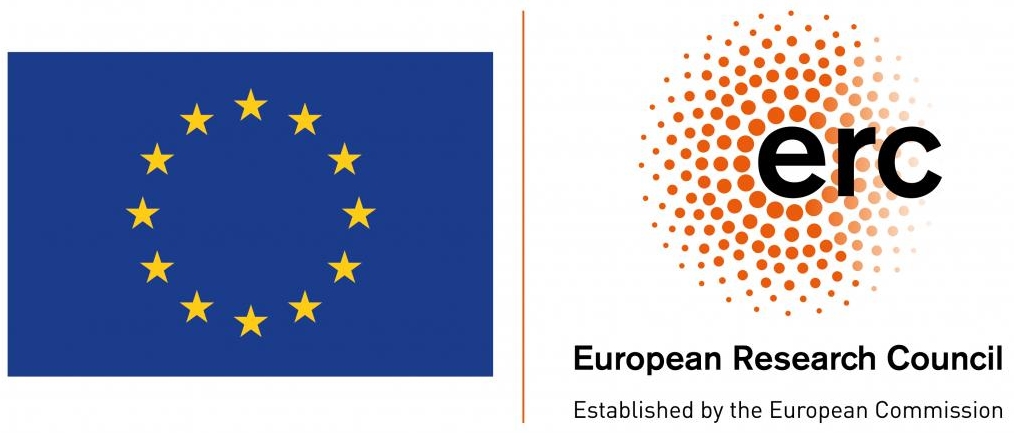}
This project has received funding from the European Research Council (ERC) under the European Union’s Horizon 2020 research and innovation programme (grant agreement No. 882500). The authors gratefully acknowledge the Gauss Centre for Supercomputing e.V. (\url{www.gauss-centre.eu}) for funding this project by providing computing time on the GCS Supercomputer SuperMUC-NG at Leibniz Supercomputing Centre (\url{www.lrz.de}).}%

\nolinenumbers %

\EventEditors{John Q. Open and Joan R. Access}
\EventNoEds{2}
\EventLongTitle{42nd Conference on Very Important Topics (CVIT 2016)}
\EventShortTitle{CVIT 2016}
\EventAcronym{CVIT}
\EventYear{2016}
\EventDate{December 24--27, 2016}
\EventLocation{Little Whinging, United Kingdom}
\EventLogo{}
\SeriesVolume{42}
\ArticleNo{23}

\newif\iffigures
\figurestrue

\usepackage[ruled,vlined,linesnumbered,norelsize]{algorithm2e}
\usepackage{accents}
\usepackage{mathtools}
\usepackage{amsthm}
\usepackage{bbm}
\usepackage{enumitem}

\SetKwFunction{Sort}{Sort}
\SetKwFunction{RQuick}{RQuick}
\SetKwFunction{LcpRQuick}{LCP-RQuick}
\SetKwFunction{LcpMerge}{LCP-Merge}
\SetKwFunction{Merge}{Merge}
\SetKwFunction{LcpLowerBound}{LCP-LowerBound}
\SetKwFunction{DetAssign}{DeterministicAssign}
\SetKwFunction{AlltoallC}{ColumnAlltoall}
\SetKwFunction{Hash}{Hash}
\SetKwFunction{FindDuplicates}{FindDuplicates}
\SetKwFunction{SendRec}{SendRec}
\SetKwProg{Fn}{Function}{}{}
\SetCommentSty{textit}
\SetFuncSty{texttt}
\SetKw{KwAnd}{and}
\SetKw{KwOr}{or}
\SetKw{KwDownTo}{downto}

\newcommand{\encoding}[2]{\textsc{enc}\left(#1, #2\right)}

\newcommand{\bigO}[1]{\mathcal{O}(#1)}

\newcommand{\lcpMath}[2]{\textup{lcp(}#1, #2\textup{)}}
\newcommand{\distMath}[1]{\textsc{dist}(#1)}
\newcommand{\expectedValue}[1]{\mathbb{E}[#1]}
\newcommand{\distApproxMath}[1]{\textsc{dist}_{\approxeq}(#1)}

\newcommand{\Bucket}{\mathcal{B}}
\newcommand{\StringSet}[1]{\mathcal{#1}}
\newcommand{\StringSetIndex}[2]{\mathcal{#1}_{#2}}

\newcommand{\maketh}[1]{#1^{th}}
\newcommand{\exchange}[3]{\textrm{Exch}(#1,#2,#3)}
\renewcommand{\exchange}[3]{\textrm{Exch}(#2,#3)}
\newcommand{\exchangetilde}[3]{\textrm{E}\widetilde{\textrm{xc}}\textrm{h}(#1,#2,#3)}
\renewcommand{\exchangetilde}[3]{\textrm{E}\widetilde{\textrm{xc}}\textrm{h}(#2,#3)}

\newcommand{\nmax}{n_{\textrm{max}}}
\newcommand{\mbar}{\bar{m}}
\DeclareMathOperator*{\argmin}{arg\,min}
\DeclareMathOperator{\poly}{poly}

\newcommand{\frage}[1]{#1}
\renewcommand{\frage}[1]{}

\newcommand{\fullstop}{}
\DeclarePairedDelimiter\parens{\lparen}{\rparen}
\DeclarePairedDelimiter\braces{\lbrace}{\rbrace}

\DeclarePairedDelimiter\abs{\lvert}{\rvert}
\DeclarePairedDelimiter\norm{\lVert}{\rVert}
\DeclarePairedDelimiter\floor{\lfloor}{\rfloor}
\DeclarePairedDelimiter\ceil{\lceil}{\rceil}

\newcommand{\uhat}{\underaccent{\check}}
\def\lmin{\ensuremath{\uhat\ell}}
\def\lmax{\ensuremath{\hat\ell}}
\def\dmax{\ensuremath{\hat{d}}}

\def\davg{\ensuremath{\bar{d}}}

\def\nmax{\ensuremath{\hat{n}}}

\def\nsquig{\widetilde{n}}

\def\Nsquig{\widetilde{N}}
\def\Dsquig{\widetilde{D}}
\def\recursionlevel{t}

\RequirePackage[usenames,dvipsnames,prologue]{xcolor}

\usepackage{latexsym}
\usepackage{siunitx}
\usepackage{datetime2}
\usepackage{csquotes}
\usepackage{booktabs}
\usepackage{threeparttablex}

\DeclareSIUnit{\nothing}{\relax}

\usepackage{tikz,pgfplots}
\usepackage{pgfmath}
\usetikzlibrary{math,fit,calc,patterns,plotmarks}
\usetikzlibrary{decorations.pathmorphing}
\usetikzlibrary{decorations.pathreplacing,calligraphy}
\usepgfplotslibrary{statistics,groupplots}
\usepgfplotslibrary{colorbrewer}

\usepackage{pgfext}

\def\ms{\texttt{MS}}
\def\pdms{\texttt{PDMS}}

\def\RQuick{\texttt{RQuick}}
\def\LcpRQuick{\texttt{RQuick+}}
\def\psV{$\text{pS}^5$}

\pgfplotsset{compat=1.18}

\definecolor{my-dark-red}{RGB}{183, 28, 28}
\definecolor{my-red}{RGB}{244,67,54}
\definecolor{my-pink}{RGB}{233,30,99}
\definecolor{my-purple}{RGB}{156,39,176}
\definecolor{my-deep-purple}{RGB}{103,58,183}
\definecolor{my-indigo}{RGB}{63,81,181}
\definecolor{my-blue}{RGB}{33,150,243}
\definecolor{my-light-blue}{RGB}{3,169,244}
\definecolor{my-cyan}{RGB}{0,188,212}
\definecolor{my-teal}{RGB}{0,150,136}
\definecolor{my-green}{RGB}{76,175,80}
\definecolor{my-light-green}{RGB}{139,195,74}
\definecolor{my-lime}{RGB}{205,220,57}
\definecolor{my-yellow}{RGB}{255,235,59}
\definecolor{my-amber}{RGB}{255,193,7}
\definecolor{my-orange}{RGB}{255,152,0}
\definecolor{my-deep-orange}{RGB}{255,87,34}
\definecolor{my-brown}{RGB}{121,85,72}
\definecolor{my-grey}{RGB}{158,158,158}
\definecolor{my-blue-grey}{RGB}{96,125,139}
\definecolor{my-lipics-grey}{rgb}{0.6,0.6,0.61}

\pgfplotscreateplotcyclelist{colorListUpToThree}{
  {my-light-blue,densely dashed,mark=o,mark options={solid}},
  {my-light-green,densely dashed,mark=triangle,mark options={solid}},
  {my-blue,mark=*},
  {my-green,mark=triangle*},
  {my-indigo,mark=diamond*},
  {my-teal,mark=square*},
  {my-amber,densely dashed,mark=pentagon,mark options={solid}},
  {my-deep-orange,mark=pentagon*}
}

\pgfplotscreateplotcyclelist{colorListUpToTwoWithPSV}{
  {my-light-blue,densely dashed,mark=o,mark options={solid}},
  {my-light-green,densely dashed,mark=triangle,mark options={solid}},
  {my-blue,mark=*},
  {my-green,mark=triangle*},
  {my-amber,dashed,mark=pentagon,mark options={solid}},
  {my-deep-orange,mark=pentagon*},
  {my-brown,densely dotted,mark=star,mark options={solid}}
}

\pgfplotscreateplotcyclelist{colorListUpToTwoWithoutPSV}{
  {my-light-blue,densely dashed,mark=o,mark options={solid}},
  {my-blue,mark=*},
  {my-light-green,densely dashed,mark=triangle,mark options={solid}},
  {my-green,mark=triangle*},
  {my-amber,dashed,mark=pentagon,mark options={solid}},
  {my-deep-orange,mark=pentagon*}
}

\pgfplotscreateplotcyclelist{colorListUpToTwoRealWorld}{
  {my-light-blue,densely dashed,mark=o,mark options={solid}},
  {my-blue,mark=*},
  {my-light-green,densely dashed,mark=triangle,mark options={solid}},
  {my-green,mark=triangle*}
}

\pgfplotscreateplotcyclelist{colorListBarPlot}{
  {my-deep-purple,fill=my-deep-purple!50,mark=none},
  {my-blue,fill=my-blue!50,mark=none},
  {my-green,fill=my-green!50,mark=none},
  {my-orange,fill=my-orange!50,mark=none},
  {my-brown,fill=my-brown!50,mark=none}
}

\pgfplotsset{
  major grid style={thin,dotted,my-grey},
  minor grid style={thin,dotted,my-grey},
  ymajorgrids,
  yminorgrids,
  axis lines*=left,
  xlabel near ticks,
  ylabel near ticks,
  axis lines*=left,
  every axis/.append style={
    line width=0.7pt,
    tick style={
      line cap=round,
      thin,
      major tick length=4pt,
      minor tick length=2pt,
    },
  },
  legend style={
    line width=0.25pt,
    /tikz/every even column/.append style={column sep=3mm,black},
    /tikz/every odd column/.append style={black},
  },
}

\pgfplotscreateplotcyclelist{phaselist}{%
    blue,fill=blue!30!white,mark=none\\
    red,fill=red!30!white,mark=none\\
    teal!60!black,fill=teal!50!white,mark=none\\
    brown!60!black,fill=brown!30!white,mark=none\\
    black,fill=gray,mark=none\\
}

\pgfplotscreateplotcyclelist{customcolors}{
    Red\\
    BlueGreen\\
    LimeGreen\\
    Plum\\
    PineGreen\\
    Bittersweet\\
    NavyBlue\\
    Apricot\\
    VioletRed\\
}

\pgfplotscreateplotcyclelist{custommarks}{%
    mark=o\\
    mark=x\\
    mark=diamond\\
    mark=star\\
    mark=triangle\\
    mark=+\\
    mark=square\\
    mark=Mercedes star\\
    mark options={rotate=180},mark=triangle\\
}

\pgfplotsset{
    major grid style={thin,dotted,color=black!50},
    minor grid style={thin,dotted,color=black!50},
    grid,
    cycle list/Set1-5,
    cycle multiindex* list={
        custommarks\nextlist
        customcolors\nextlist
    },
    every axis/.append style={
        line width=0.5pt,
        tick style={
            line cap=round,
            thin,
            major tick length=4pt,
            minor tick length=2pt,
        },
    },
    every axis y label/.style={
            at={(ticklabel cs:0.5)},
            rotate=90,
            anchor=near ticklabel,
        },
    custom ybar legend/.style={
        legend image code/.code={
            \draw[##1] (0cm,-0.6ex) rectangle + (0.4cm,1.5ex);
        },
    },
}

\newcommand\irregularline[2]{%
  let \n1 = {rand*(#1)} in +(0,\n1)
  \foreach \a in {0.15,0.3,...,#2}{
    let \n1 = {rand*(#1)} in -- +(\a,\n1)
  } 
  let \n1 = {rand*(#1)} in -- +(#2,\n1)
}

\tikzset{
  fitting node/.style={
    inner sep=0pt,
    fill=none,
    draw=none,
    reset transform,
    fit={(\pgf@pathminx,\pgf@pathminy) (\pgf@pathmaxx,\pgf@pathmaxy)}
  },
  reset transform/.code={\pgftransformreset}
}

\begin{document}

\maketitle

\begin{abstract}
	String sorting is an important part of tasks such as building index data structures.
	Unfortunately, current string sorting algorithms do not scale to massively parallel distributed-memory machines since they either have latency (at least) proportional to the number of processors $p$ or communicate the data a large number of times (at least logarithmic).
	We present practical and efficient algorithms for distributed-memory string sorting that scale to large $p$.
	Similar to state-of-the-art sorters for atomic objects, the algorithms have latency of about $p^{1/k}$ when allowing the data to be communicated $k$ times.
	Experiments indicate good scaling behavior on a wide range of inputs on up to \num{49152} cores. Overall, we achieve speedups of up to $5$ over the current state-of-the-art distributed string sorting algorithms.
\end{abstract}

\section{Introduction}
Sorting strings is a fundamental building block of many important string-processing tasks such as the construction of index data structures for databases and full-text phrase search~\cite{ferragina1999string, DBLP:journals/jacm/KarkkainenSB06, nong2013practical}.
The problem differs from \emph{atomic} sorting---where keys are treated as indivisible objects that can be compared in constant time.
Strings on the other hand can have variable lengths and the time needed to compare two strings depends on the length of their longest common prefix.
Therefore, string sorting algorithms try to avoid the comparison of whole strings.
Instead, they only inspect the \emph{distinguishing prefixes} of the strings, i.e., the characters needed to establish the global ordering, only once.
The sum of the lengths of all distinguishing prefixes is usually denoted by $D$.
The lower bound for sequential string sorting based on character comparisons is $\Omega(n\log n + D)$ with existing algorithms matching this bound \cite{DBLP:conf/soda/BentleyS97}.

\vspace{-.1cm}
\subparagraph{Related Work.}
There exists extensive research on string sorting in the sequential setting.
For a systematic overview, we refer to \cite{DBLP:phd/dnb/Bingmann18, DBLP:conf/spire/KarkkainenR08, DBLP:conf/wea/SinhaW08}.
We focus on parallel sorting algorithms.
Let \(n\) be the total number of strings and \(N\) be the total number of characters.
See \cref{tab:symbols_and_abbreviations} for a list of abbreviations used in this paper.
For the PRAM model, there are (comparison-based) algorithms solving the string sorting problem in $\mathcal{O}(n\log n + N)$ work and $\mathcal{O}(\log^2 n/\log\log n)$ time~\cite{DBLP:journals/tcs/JaJaRV96}.
For integer alphabets, Hagerup proposes an algorithm with $\mathcal{O}(N\log N)$ work and running time in $\mathcal{O}(\log N/\log\log N)$~\cite{DBLP:conf/stoc/Hagerup94}.
Ellert~et~al. propose a framework for PRAM string sorting algorithms that makes the (work) complexity depend on $D$ instead of $N$ by increasing the time complexity by a logarithmic factor in the length of the longest distinguishing prefix~\cite{DBLP:conf/europar/Ellert0S20}.
String sorting is also considered in other parallel settings, e.g., on the GPU~\cite{neelima2014string}.

While the problem has been extensively studied in the sequential and (shared-memory) parallel setting, we are only aware of the following results in the distributed-memory setting.
Bingmann~et~al. present the state-of-the-art distributed string sorting algorithms \cite{DBLP:conf/ipps/Bingmann0S20}.
The first one follows the standard distributed-memory merge sort scheme (local sorting, partitioning, message exchange, and merging) \cite{DBLP:journals/jpdc/VarmanSIR91}.
Every step is augmented with string-specific optimizations, e.g., LCP-compression and LCP-aware merging \cite{ng2008merging, DBLP:conf/ipps/Bingmann0S20, DBLP:journals/algorithmica/BingmannES17}, see \cref{sec:preliminaries}.
The second algorithm is more communication-efficient, as it only sorts approximations of the distinguishing prefixes using the first algorithm.
They also adapt the distributed hypercube quicksort algorithm \cite{axtmann2015practical, axtmann2017robust} to variable-length keys (without string-related optimizations).
These algorithms improve the first dedicated distributed string sorting algorithm by Fischer~et~al.~\cite{DBLP:conf/alenex/0001K19}.

However, the algorithms are only efficient for very small or large inputs, as they have a prohibitively high communication volume (hypercube quicksort) or do not scale to the largest available machines due to their latency which is (at least) proportional to the number of processors $p$.
For example, the merge sort algorithms do not scale to the largest available machines, as the partitioning step dominates the running time unless $n=\Omega(p^2 \log p)$.

\crefname{theorem}{Thm.}{Thm.}
\vspace{-.1cm}
\subparagraph{Our Contribution.}
Our new algorithms close this gap by providing a viable trade-off between latency and communication volume.
We introduce a multi-level approach to Bingmann~et~al.'s algorithms, i.e., we use $k$ levels where processor groups work on independent sorting problems.
The resulting {\em multi-level mergesort} has internal work and communication volume close to \(N\) for each level (\cref{thm:multi_level:simple:complexity}).
This is significantly improved with {\em prefix-doubling mergesort} 
where internal work and communication volume per level are close to $D$ (\cref{thm:multi_level:pdms}).
As side results of independent interest, we present an improved {\em hypercube quicksort for strings} 
(\cref{thm:rquick-runningtime}) and a {\em multi-level distributed single-shot Bloom filter} (\cref{thm:k-level-bloomfilter}).
While the idea to combine the best string sorting and atomic sorting algorithm is simple in principle, we view the analysis of a quite complicated overall algorithm in a realistic model of distributed-memory computing as a significant contribution. In particular,
because this guides an efficient, highly scalable implementation.
In an experimental evaluation, on up to \num{49152} cores, the multi-level algorithms are up to $5\times$ faster than the single-level ones.

\crefname{theorem}{Theorem}{Theorems}

\section{Preliminaries}
\label{sec:preliminaries}
\vspace{-.1cm}
\subparagraph{Machine Model and Communication Primitives.}
We assume a distributed-memory machine model consisting of $p$ processing elements (PEs) allowing single-ported point-to-point communication.
The cost of exchanging a message of $h$ bits between any two PEs is $\alpha + \beta h$, where $\alpha$ accounts for the message start-up overhead and $\beta$ quantifies the time to exchange one bit.
Let $h$ be the maximum number of bits a PE sends or receives, then collective operations \emph{broadcast}, \emph{prefix sum}, \emph{(all-)reduce}, and \emph{(all-)gather} can be implemented in time $\mathcal{O}(\alpha \log p + \beta h)$~\cite{DBLP:books/sp/SandersMDD19}.
For personalized all-to-all communication, there is a trade-off between communication volume and start-up latency.
When data is delivered directly, we obtain a time complexity in $\mathcal{O}(\alpha p + \beta h)$, while using a maximum degree of indirection yields $\mathcal{O}(\alpha \log p + \beta h\log p )$.
In our analysis, we resort to a more abstract view introduced by Axtmann~et~al.---a black box data exchange function $\exchange{P}{h}{r}$ yielding the time complexity of exchanging data%
, when each PE sends or receives at most $h$ bits in total from at most $r$ PEs \cite{axtmann2015practical}.
We find $\alpha r + \beta h$ as a lower bound for the time complexity of $\exchange{P}{h}{r}$, and although matching upper bounds are not known, there are indications suggesting that one can come close to this (see \cite{axtmann2015practical} for a brief discussion).
We use $\exchangetilde{P}{h}{r} = (1+o(1))\exchange{P}{h}{r}$ to sum multiple exchanges by the dominant one.

\begin{table}[t]
	\centering
	\caption[List of used symbols]{Symbols used in this paper.}
	\label{tab:symbols_and_abbreviations}
	\begin{tabular}{ll}
		\multicolumn{2}{c}{\textbf{Machine Model}}            \\
		\toprule
		$ p$       & number of processing elements (PEs)      \\
		$ r$       & number of groups in each recursion level \\
		$ k$       & number of recursion levels               \\
		$ \alpha $ & message start-up latency                 \\
		$ \beta $  & time to communicate a bit                \\
		\bottomrule
	\end{tabular}
	\begin{tabular}{ll}
		\multicolumn{2}{c}{\textbf{String Properties}}    \\
		\toprule
		$n$     & total number of strings                 \\
		$N$     & total number of characters              \\
		$\lmax$ & length of longest string                \\
		$\lmin$ & length of shortest string               \\
		$\dmax$ & length of longest distinguishing prefix \\
		\bottomrule
	\end{tabular}
\end{table}
\vspace{-.25cm}
\subparagraph{String Properties and Input Format.}
The input to our algorithms is a string array $S = [s_0, s_1, \ldots, s_{n-1}]$ consisting of $n = \abs{S}$ unique strings.
A string $s$ of length $\ell$ is a sequence of characters from an alphabet $\Sigma$ with $s = [s[0], \ldots, s[\ell-2],\bot]$ where $\bot \notin \Sigma$ is a sentinel character; $\lmax$ denotes the length of the longest string in $S$.
By $N = \norm{S}$, we denote the total number of characters in $S$.
The $\ell$-prefix of a string $s$ are the first $\ell$ characters of $s$.
The \emph{longest common prefix (LCP)} of two strings $s \neq t$ is the prefix of $s$ with length $\lcpMath{s}{t} = \argmin{s[i] \neq t[i]}$.
For a sorted string array, the corresponding LCP array $[\bot, h_1, h_2, \ldots, h_{n-1}]$ contains the LCP values $h_i = \lcpMath{s_{i-1}, s_i}$.
For sorting the string array, we do not necessarily have to look at all characters in $S$.
The \emph{distinguishing prefix} of a string $s$ (with length $\distMath{s}$) are the characters that need to be inspected to rank $s$ in $S$.
The sum of the lengths of all distinguishing prefixes of $S$ is denoted by $D$.
By $\dmax$, we denote the length of the longest distinguishing prefix.
A string array is usually represented as an array of pointers referring to the start of the corresponding character sequence. This allows for moving or swapping strings in constant time.
The concatenation of all the character sequences forms the \emph{character} array $\mathcal{C}\mathcal({S})$ with $\abs{\mathcal{C}\mathcal({S})} = \norm{\mathcal{S}}$.

In our distributed setting, we assume that each PE $i$ obtains a local subarray $S_i$ of $S$ as input such that $S$ is the concatenation of all local string arrays $S_i$.
Furthermore, we assume the input to be well-balanced, i.e., $\abs{S_i} = \Theta(n/p)$ and $\norm{S_i} = \Theta(N/p)$.

\vspace{-.25cm}
\subparagraph{Algorithmic Building Blocks.}
\label{sec:algorithmic_building_blocks}
We make use of an \emph{$r$-way LCP loser tree} to merge $r$ sorted sequences of in total $m$ strings augmented with LCP values in $\mathcal{O}(m\log r +D)$ time~\cite{DBLP:journals/algorithmica/BingmannES17}.
While merging, we update the LCP values of the resulting merged string sequence.
Furthermore, we use LCP compression, i.e., we send the longest common prefix of two consecutive strings (in a sorted sequence) only once.
While being very useful for many inputs in practice, LCP compression cannot substantially reduce the communication volume in the worst case \cite{DBLP:conf/ipps/Bingmann0S20, schimek2019distributed}.

Robust Hypercube Quicksort (RQuick) is a sorting algorithm initially proposed for atomic sorting \cite{axtmann2017robust, DBLP:phd/dnb/Axtmann21} that has been adapted to handle strings~\cite{DBLP:conf/ipps/Bingmann0S20}.
RQuick only works for a number of PEs that is a power of two.
Hence, with $d = \floor{\log p}$, PEs with index $i \ge 2^d$ send their data to a PE within the \(d\)-dimensional hypercube.
Then the input data is randomly permuted among the PEs to make imbalances less likely.
Subsequently, each PE sorts its local data.
Then the actual $d$ hypercube (quick)sorting rounds are executed, where the input is recursively partitioned into sub-hypercubes using a pivot element.
Hence, the elements are communicated at least $d = \floor{\log p}$ times.

\begin{theorem}[String RQuick,{\cite[Theorem 1]{DBLP:conf/ipps/Bingmann0S20}}]
	\label{thm:rquick-runningtime}
	If all input strings are unique, RQuick runs in time
	\(
	\mathcal{O}\left( \lmax \frac{n}{p} \log n + \alpha \log^2p + \beta \left(\frac{n}{p} \lmax \log p + \lmax \log^2p   \right)\log{\sigma} \right)
	\) with probability $\ge 1 - p^{-c}$ for any constant $c>0$.
\end{theorem}

RQuick can be easily improved by using a string sorting algorithm for local sorting and (LCP-aware) merging.
Our implementation \LcpRQuick\ uses these optimizations.
Therefore, the overall work performed by all PEs for these steps is $\mathcal{O}(n\log n + D)$.

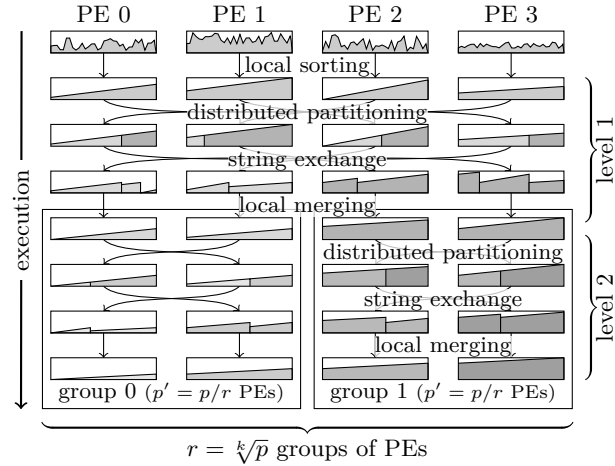
\begin{figure}[t]
	\centering
	\begin{tikzpicture}[scale=0.275,font=\footnotesize]

    \tikzmath{ \dx = 5; \dy = -1; \mx = 1.5; \my = -1.25; \dmx = \dx + \mx; \dmy = \dy + \my; }

    \foreach \y in {0, ..., 7} \foreach \x in {0, ..., 3} \fill[fill=white]
    (\x * \dmx, \y * \dmy) rectangle (\x * \dmx + \dx, \y * \dmy + \dy) node[fitting node] (c\y\x) {};

    \node foreach \x in {0,...,3} [anchor=south] at (c0\x.north) {\small{PE \x}};

    \path [draw,fill=gray!40] (c00.west) ++ (0,-0.1) {\irregularline{0.3}{5}}   -- (c00.south east) -- (c00.south west) -- cycle;
    \path [draw,fill=gray!40] (c01.west) ++ (0,+0.2) {\irregularline{0.3}{5}}   -- (c01.south east) -- (c01.south west) -- cycle;
    \path [draw,fill=gray!40] (c02.west) ++ (0,-0.05) {\irregularline{0.4}{5}}  -- (c02.south east) -- (c02.south west) -- cycle;
    \path [draw,fill=gray!40] (c03.west) ++ (0,-0.15) {\irregularline{0.15}{5}} -- (c03.south east) -- (c03.south west) -- cycle;

    \foreach \y/\x/\p/\q/\r/\cl/\cr in {
            1/0/0.1/0.7/1.0/40/40,   1/1/0.4/1.0/1.0/40/40,   1/2/0.0/0.9/1.0/40/40,  1/3/0.3/0.6/1.0/40/40,
            2/0/0.1/0.7/0.67/30/60,  2/1/0.4/1.0/0.167/30/60, 2/2/0.0/0.9/0.56/30/60, 2/3/0.3/0.6/0.67/30/60,
            4/0/0.0/0.5/1.0/30/30,   4/1/0.1/0.5/1.0/30/30,   4/2/0.6/0.9/1.0/60/60,  4/3/0.5/1.0/1.0/60/60,
            5/0/0.0/0.5/0.375/20/40, 5/1/0.1/0.5/0.6/20/40,  5/2/0.6/0.9/0.6/55/75,  5/3/0.5/1.0/0.4/55/75,
            7/0/0.0/0.25/1.0/20/20,  7/1/0.25/0.5/1.0/40/40,  7/2/0.5/0.75/1.0/55/55, 7/3/0.75/1.0/1.0/75/75
        }
        {
            \coordinate (l\y\x) at ($(c\y\x.south west)!\p!(c\y\x.north west)$);
            \coordinate (r\y\x) at ($(c\y\x.south east)!\q!(c\y\x.north east)$);
            \coordinate (m\y\x) at ($(l\y\x)!\r!(r\y\x)$);
            \coordinate (b\y\x) at ($(c\y\x.south west)!\r!(c\y\x.south east)$);

            \fill[fill=gray!\cl] (c\y\x.south west) -- (l\y\x) -- (m\y\x) -- (b\y\x) -- cycle;
            \fill[fill=gray!\cr] (c\y\x.south east) -- (b\y\x) -- (m\y\x) -- (r\y\x) -- cycle;

            \draw (l\y\x) -- (r\y\x) (b\y\x) -- (m\y\x);
        }

    \foreach \y/\x/\lx/\ly/\rx/\ry/\c in {
            3/0/0.0/0.1/0.67/0.5/30, 3/0/0.67/0.4/0.85/0.5/30, 3/0/0.85/0.0/1.0/0.15/30,
            3/1/0.0/0.15/0.4/0.5/30, 3/1/0.4/0.3/1.0/0.5/30,
            3/2/0.0/0.5/0.33/0.7/60, 3/2/0.33/0.5/1.0/0.9/60,
            3/3/0.0/0.9/0.2/1.0/60, 3/3/0.2/0.5/0.67/0.9/60, 3/3/0.67/0.5/1.0/0.6/60,
            6/0/0.0/0.0/0.375/0.25/20, 6/0/0.375/0.1/1.0/0.25/20,
            6/1/0.0/0.25/0.6/0.5/40, 6/1/0.6/0.3/1.0/0.5/40,
            6/2/0.0/0.6/0.6/0.75/55, 6/2/0.6/0.5/1.0/0.70/55,
            6/3/0.0/0.75/0.4/0.9/75, 6/3/0.4/0.75/1.0/1.0/75
        }
        {
            \coordinate (pA\y\x) at ($(c\y\x.south west)!\lx!(c\y\x.south east)$);
            \coordinate (pB\y\x) at ($(pA\y\x)-\ly*(0, \dy)$);
            \coordinate (pD\y\x) at ($(c\y\x.south west)!\rx!(c\y\x.south east)$);
            \coordinate (pC\y\x) at ($(pD\y\x)-\ry*(0, \dy)$);

            \draw[fill=gray!\c] (pA\y\x) -- (pB\y\x) -- (pC\y\x) -- (pD\y\x) -- cycle;
        }

    \foreach \y/\z in {0/1,3/4,6/7} \foreach \x in {0,...,3} \draw[->] (c\y\x.south) -- (c\z\x.north);

    \foreach \i/\j in {3/0, 2/1} \draw[out=225,in=45,looseness=0.25,->] (c1\i.south) to (c2\j.north);
    \foreach \i/\j in {3/0, 2/1} \draw[out=315,in=135,looseness=0.25,->] (c1\j.south) to (c2\i.north);

    \foreach \i/\j in {3/0, 2/1} \draw[out=225,in=45,looseness=0.25,->] (c2\i.south) to (c3\j.north);
    \foreach \i/\j in {3/0, 2/1} \draw[out=315,in=135,looseness=0.25,->] (c2\j.south) to (c3\i.north);

    \foreach \i/\j in {1/0, 3/2} \draw[out=225,in=45,looseness=0.5,->] (c5\i.south) to (c6\j.north);
    \foreach \i/\j in {1/0, 3/2} \draw[out=315,in=135,looseness=0.5,->] (c5\j.south) to (c6\i.north);

    \foreach \i/\j in {1/0, 3/2} \draw[out=225,in=45,looseness=0.5,->] (c4\i.south) to (c5\j.north);
    \foreach \i/\j in {1/0, 3/2} \draw[out=315,in=135,looseness=0.5,->] (c4\j.south) to (c5\i.north);

    \draw ($(c40.north west) + (-0.4, 0.4)$) rectangle ($(c71.south east) + (0.4, -1.4)$) node[fitting node] (group0) {};
    \draw ($(c42.north west) + (-0.4, 0.4)$) rectangle ($(c73.south east) + (0.4, -1.4)$) node[fitting node] (group1) {};
    \node foreach \i in {0,1} [anchor=south, yshift=-1pt] at (group\i.south) {group \i{} {\scriptsize ($p'=p/r$ PEs)}};

    \foreach \i/\j/\l/\r/\phase in {
            0/1/0/3/local sorting,
            1/2/0/3/distributed partitioning,2/3/0/3/string exchange,3/4/0/3/local merging,
            4/5/2/3/distributed partitioning,5/6/2/3/string exchange,6/7/2/3/local merging
        }
    \node[inner sep=0.025pt, fill=white, opacity=0.75, text opacity=1,yshift=-.5pt] at ($(c\i\l.south west)!0.5!(c\j\r.north east)$) {\phase};

    \foreach \y in {0, ..., 7} \foreach \x in {0, ..., 3} \draw (c\y\x.south west) rectangle (c\y\x.north east);

    \draw[thick,->] ($(c00.north west) + (-1.4, 0)$) -- ($(group0.south west) + (-1, 0)$) node[midway,rotate=90,fill=white] {\small{execution}};

    \draw[decoration={calligraphic brace,raise=5pt,amplitude=4pt},decorate,thick] ($(c13.north east) + (0.4, 0.0)$) -- node[anchor=north, inner sep=2ex, rotate=90] {\small level 1} ($(c43.north east) + (0.4, -0.2)$);
    \draw[decoration={calligraphic brace,raise=5pt,amplitude=4pt},decorate,thick] ($(c43.south east) + (0.4, 0.2)$) -- node[anchor=north, inner sep=2ex, rotate=90] {\small level 2} ($(c73.south east) + (0.4, -0.0)$);
    \draw[decoration={calligraphic brace,raise=5pt,amplitude=4pt,mirror},decorate,thick] (group0.south west) -- node[below=7.5pt] {\small $r=\sqrt[\leftroot{2}k]{p}$ groups of PEs} (group1.south east);

\end{tikzpicture}
	\caption{Overview of the main steps in the multi-level string sorting scheme with \(k=2\) levels.}
	\label{fig:merge-sort-scheme}
\end{figure}

\section{Multi-Level String Sorting}
\label{sec:multi_level_string_sorting}
Our \emph{multi}-level \textbf{m}erge \textbf{s}ort (MS) approach adapts the ideas of Axtmann~et~al. for multi-level \emph{atomic} sorting \cite{axtmann2015practical, axtmann2017robust} to string sorting.
We recursively split PEs into groups each of which solves an independent (string) sorting problem.
There are $k$ levels of recursion with arbitrary splitting factors between levels.
To obtain the single-level variant~\cite{DBLP:conf/ipps/Bingmann0S20} of the algorithm choose $k=1$.
To simplify the analysis, we assume that $p$ can be perfectly subdivided into $r$ groups of $p/r$ consecutive PEs, i.e., on the first level the $\maketh{j}$ group consists of PEs $jp/r, \ldots, (j+1)p/r-1$.
Furthermore, we generally assume approximately equal factors on each level, i.e., $r=\Theta\parens{\sqrt[\leftroot{2}k]{p}}$ which implies $p=\Theta\parens{r^k}$, and that the final levels splits the PEs into groups of size~$1$.
The algorithm consists of a one-time initialization and a recursive phase which is invoked $k$ times.
\cref{fig:merge-sort-scheme} provides an overview.

\noindent
{\textbf{Initialization:}}
The input is sorted locally.
On each PE $i$, sort the local input array $\mathcal{S}_i$.
The LCP array can be obtained as a by-product of sorting.

\noindent
{\textbf{Recursion:}}
A global order is established recursively.
Initially, all string arrays $\mathcal{S}_i$ must be locally sorted.
On each level of recursion, we have $p'$ PEs and $r$ groups of size $p''=p'/r$.

\begin{enumerate}[nolistsep,nosep]
	\item \label{alg:multi_level:simple:partition}
	      \textbf{Distributed Partitioning:}
	      Globally determine $r-1$ splitter strings $f_j$ and on each PE $i$ compute local \emph{buckets} $\mathcal{B}_i^0, \dots,
          \mathcal{B}_i^{r-1}$ with $\mathcal{B}_i^j = \{ s \in S_i \mid f_{j} < s \le f_{j+1} \}$ for $j \in \braces{0, \ldots, r-1}$ using sentinels $f_{0} = -\infty$ and $f_{r} = \infty$.

	\item  \label{alg:multi_level:simple:exchange}
	      \textbf{String Assignment and Exchange:}
	      On PE $i$, the strings
	      in bucket $\mathcal{B}_i^j$ are assigned to PEs belonging to group $j$. By $\mathcal{B}^j = \bigcup_i \mathcal{B}_i^j$ we denote the union of all strings assigned to group $j$.
	      Then, all strings and LCP values are exchanged using direct messaging.

	\item  \label{alg:multi_level:simple:merge}
	      \textbf{Local LCP-aware Merging:}
	      On PE $i$, the received string sequences are merged to obtain locally sorted string arrays
	      $\mathcal{O}_i$ (using the LCP values).
	      We also update the LCP values for $\mathcal{O}_i$ during merging.
	      We then set $p' \leftarrow p'/r$ and $S_i \leftarrow O_i$ in the subsequent recursive step.
\end{enumerate}

\subsection{Distributed Partitioning}
Due to the multidimensionality of the string sorting problem, determining balanced partitions is far more challenging than in atomic sorting.
Some of the steps of our merge sort algorithm depend on the number of strings in the local string array while others depend on the number of characters or the size of the distinguishing prefix.
We therefore adapt \emph{string-based} and \emph{character-based} partitioning~\cite{DBLP:conf/ipps/Bingmann0S20} to our multi-level approach.
These schemes bound the number of strings and characters, respectively.
The general approach is to draw and globally sort a number of samples on each PE.
Then, $r-1$ splitters $f_j$ are chosen from the global sample array.
Since all strings are locally sorted before the partitioning step, we can make use of a \emph{regular} sampling approach \cite{DBLP:journals/pc/LiLSS93, DBLP:journals/jpdc/ShiS92} in which the samples are drawn equidistantly. %

\subsubsection{String-Based Partitioning}
On the $\maketh{\recursionlevel}$ recursion level, there are $r^{\recursionlevel-1}$ groups of PEs working on independent sorting problems, see \cref{fig:merge-sort-scheme}.
We now describe the partitioning process from the point of view of one such group.
Let $\StringSet{S'}$ be the concatenation of the local string arrays of the PEs of one such group of size $p'$.
Unlike for single-level MS, now, we cannot assume an equal number of strings on each PE as from the second level of recursion on, this number is subject to the result of a previous partitioning round which is not exact.
For multi-level MS, string-based partitioning consists of the following two steps:\\
\vspace{-.125cm}
\begin{description}[nosep]
	\item[Local Sampling:] Let $v > 0$ be the \emph{sampling} factor.
	      In total, there will be $p'(v+1)$ samples drawn from the local string arrays.
	      To simplify the discussion, we assume $\abs{\StringSet{S'}}$ to be divisible by $p'(v+1)$.
	      Let $\omega = \abs{\StringSet{S'}}/(p'(v+1))$.
	      PE $i$ then draws $\ceil{\abs{\StringSetIndex{S}{i}}/\omega} - 1$ samples $\mathcal{V}_i$ from its local string array spaced as evenly as possible.
	      $\mathcal{V}$ is the union of all local samples.
	      If $\abs{\mathcal{V}} < p'(v+1)$, then the first $p'(v+1) - \abs{\mathcal{V}}$ PEs draw one additional sample.
	      This ensures that at most \(\omega\) strings are between two local samples on each PE.
	      Also, the global number of samples is a multiple of $p'$ and of $r$ as we find $p' = r^{k+1-\recursionlevel}$ on recursion level $\recursionlevel$.
	\item[Splitter Computation:] The samples $\mathcal{V}$ are globally sorted using hypercube quicksort.
	      Then $r-1$ splitters $f_j = \mathcal{V}[j|\mathcal{V}|/r -1]$ for $0 < j < r$ are determined using a prefix-sum.
	      Subsequently, the $r-1$ splitters are communicated to all PEs using an all-gather operation.
\end{description}

A sampling factor $v = \Theta(r)$ yields a maximum number of strings per bucket in $\Theta(\abs{\mathcal{S}'}/r)$.
This is shown in detail by using a generalization of the sample density lemma from \cite{DBLP:conf/ipps/Bingmann0S20} in \cref{appendix:sec:string_based_partitioning}. %
The proof of the following \cref{lemma:string-based-sampling} can also be found in \cref{appendix:sec:string_based_partitioning}.

\begin{lemma}
	\label{lemma:string-based-sampling}
	On recursion level $\recursionlevel$ with $r = \sqrt[k]{p}$ in step~\ref{alg:multi_level:simple:partition} of multi-level MS, string-based regular sampling with sampling factor $v$ yields a maximum bucket size of
	\(
    |\StringSet{B}^j| \le \left(1 + \frac{r}{v} \right)^{\recursionlevel} \frac{n}{r^{\recursionlevel}}.
	\)
\end{lemma}

The term $(1+r/v)^{\recursionlevel}$ signifies that the imbalance between the buckets  multiplies with each level of recursion.
Therefore, we need to choose $v=\Theta(kr)$ for any number of \(k\) levels to keep the term asymptotically constant during the entire sorting process.
For single-level MS, $v = \Theta(r) = \Theta(p)$ samples per PE are sufficient. %
Note the large difference between drawing $kr$ samples per PE (on average) here and $r^k = p$ samples in the single-level case.
Using the assignment strategy described in \cref{section:assignment-strategies}, which equally distributes the strings over the PEs in each group, we arrive at \cref{thm:string-based-sampling}.

\begin{theorem}[String-Based Sampling]
	\label{thm:string-based-sampling}
	Using a sampling factor of $v = \Theta(kr) = \Theta(k\sqrt[k]{p})$ the number of strings per PE is in $\bigO{n/p}$ in each level of the algorithm.
\end{theorem}

\subsubsection{Character-Based Sampling}
We generalize \emph{character}-based regular sampling \cite{DBLP:conf/ipps/Bingmann0S20} to our multi-level approach to achieve tighter bounds on the number of characters per PE than the conservative $\mathcal{O}(\lmax n/p)$.
Now, each PE of the considered group draws $\ceil{\norm{\StringSet{S}_i}/ \omega'} - 1$ equally spaced samples from its character array with sampling distance $\omega'= \norm{\StringSet{S'}}/(p'(v+1))$.
To arrive at the final string samples, we shift the sampled character positions by at most $\lmax - 1$ characters to the beginning of a string.
If the total number of samples is smaller than $p'(v+1)$, the first PEs draw one additional sample.

In \cref{lemma:multi_level:complexity:character_bucket_size}, we give bounds on the number of characters per bucket over the course of our algorithm when using character-based partitioning.
This is shown in detail by using a generalization of the (character-)~sample density lemma~\cite{DBLP:conf/ipps/Bingmann0S20} in \cref{appendix:sec:character_based_partitioning}. %
\begin{lemma}%
	\label{lemma:multi_level:complexity:character_bucket_size}
	On recursion level $\recursionlevel$ with $r=\sqrt[\leftroot{2}k]{p}$ in
	step~\ref{alg:multi_level:simple:partition} of multi-level MS using character-based regular
	sampling with a sampling factor of $v$, each bucket contains at most \(\norm{\mathcal{B}^j}\leq \parens*{1+\frac{r}{v}}^{\recursionlevel}\parens*{\frac{N}{r^\recursionlevel}	+ \recursionlevel \parens*{1+\frac{v+1}{r}}\frac{p}{r^{\recursionlevel-1}}\lmax}\) characters.
\end{lemma}
The additional term depending on $\lmax$ stems from shifting the sampled positions to the beginning of strings.
By distributing the characters equally over the PEs in each group up to additional $\mathcal{O}(\lmax)$ characters (see \cref{section:assignment-strategies}), we can limit the maximum number of characters per PE in \cref{theorem:character-based-sampling}.%
\begin{theorem}[Character-Based Sampling]
	\label{theorem:character-based-sampling}
	Using a sampling factor in $\Theta(kr)$ the maximum number of characters per PE in each level is in
	\(
	\mathcal{O}\parens*{\frac{N}{p} + k^2r\lmax}.
	\)
\end{theorem}
For the single-level case with $k=1$, this is equivalent to $O(N/p + p\lmax)$ which is the bound in the original algorithm \cite{DBLP:conf/ipps/Bingmann0S20}.
For $k > 1$, we even have an improvement over the single-level algorithm.
Since we assume $k$ in $\bigO{\log p/\log\log{p}}$, we find $k^2r\lmax = \mathcal{O}(\log^2(p)\sqrt[k]{p}\lmax) = o(p\lmax)$.
This may seem counter-intuitive at first as we introduce a potential imbalance already in the first recursion level.
However, the subsequent assignment step distributes this imbalance equally over $p'/r$ PEs.

\subsection{String Assignment and Exchange}
\label{section:assignment-strategies}
We now assign the strings in bucket $\mathcal{B}^j$ to the $\maketh{j}$ (sub-)group consisting of PEs $j$,\ldots, $(j + 1)p'' - 1$ with $p'' = p'/r = r^{k - \recursionlevel}$ on the $\maketh{\recursionlevel}$ level.
The resulting assignment needs to ensure that each PE receives approximately the same amount of data.
Additionally, the number of sent and received messages per PE should be bounded by the number of groups $r$.
We generalize an approach proposed by Axtmann~et~al.~\cite{axtmann2015practical} for distributed atomic sorting to string sorting.
For the sake of simplicity, we assume that each PE contributes the same number of strings (characters) in the assignment process.
With (slightly) imbalanced data (due to the partitioning) the below-stated results hold up to a small factor.

For \emph{string-based} partitioning, we want to balance the number of strings per PE.
Since the assignment algorithm does not rely on internal properties of the elements, we can treat a string as an \emph{atomic} object and apply Axtmann~et~al.'s assignment algorithm directly.

For \emph{character-based} partitioning, we want to achieve a balanced number of characters but cannot split strings.
Therefore, we reiterate the steps of the algorithm and describe the necessary adjustments.
A local bucket $\mathcal{B}_i^j$ is \emph{small} if it contains at most $\norm{\mathcal{B}^j}/(2rp'')$ characters.
Small buckets are separately enumerated for each group $j$ using a prefix sum where each PE contributes its number of small buckets for group $j$.
The $\maketh{t}$ small bucket belonging to group $j$ is then assigned to PE $\floor*{t/r}$ of group $j$.
This way, each PE gets assigned no more than half of its final capacity and receives messages from at most $r$ different PEs.

Then, a description of each large bucket located on PE $i$ and destined for group $j$ is first sent to PE $\floor*{i/r}$ in group $j$ and each group computes a balanced assignment independently of each other.
Conceptually, this works by performing separate prefix sums over residual capacities (remaining after assigning small buckets) and sizes of unassigned buckets, in each group.
The resulting sorted sequences of integers $R$ (residual capacities) and $U$ (unassigned buckets) must then be merged such that the bucket beginning at the $\maketh{i}$ element is preceded by the PE containing the $\maketh{i}$ open slot.
A subsequence of $\langle r_i,u_j,\ldots, u_{j+h}, r_{i+1},z \rangle$ in the merged sequence of $R$ and $U$ means that the local buckets $u_{j}, \ldots, u_{j+h}$ are assigned to PE $i$.
The last bucket $u_{j+h}$ potentially needs to be split up (respecting string boundaries in the character-based assignment) and partly assigned to $r_{i+1}$ or even $r_{i+2}$ if $z = r_{i+2}$.
Since we cannot split up strings, we may end up with a PE obtaining up to $\lmax - 1$ additional characters.
Strings from one local bucket cannot be assigned on more than 3 PEs as the residual capacity on each PE is at least $\norm{\mathcal{B}^j}/(2p'')$.
Since large buckets contain more than $\norm{\mathcal{B}^j}/(2rp'')$ elements, a single PE can store at most $2r$ of them.
Hence, each PE receives $\bigO{r}$ messages.
We refer to~\cite{axtmann2015practical} for details on the actual group-local merging process.
\begin{theorem}[Bounded Assignment]
	\label{thm:bounded-assignment}
	Using the bounded assignment algorithm~\cite{axtmann2015practical}, we obtain a message assignment where each PE sends and receives $\bigO{r}$ messages
	and each PE in group~$j$ obtains $\frac{\abs{\mathcal{B}^j}}{p''}$ strings (string-based) or at most $\frac{\norm{\mathcal{B}^j}}{p''} + \lmax$ characters (character-based).
\end{theorem}

Hence, even for character-based sampling, we can find an assignment, such that the number of characters per PE remains in the bounds of \cref{theorem:character-based-sampling}.
Afterwards, each PE sends its strings according to their assignment.
The time for computing the assignment is dominated by the actual data exchange \cite{axtmann2015practical}.%
\footnote{
	For string-based assignment, this naturally transfers from the atomic sorting case.
	For character-based assignment, we have to compensate a communication volume in $\mathcal{O}(r\log{N})$.
	Assuming unique strings, we find $\lmax \ge \dmax = \Omega(\log{n}/\log{\sigma})$ and $N \le n\lmax$ and therefore $\log(N) = \mathcal{O}(\lmax \log{\sigma})$.
	Since our bound for the number of (received) characters per PE contains an imbalance of at least $r\lmax$ characters (which require an encoding of $\log{\sigma}$ bits each), the data exchange dominates the assignment also for character-assignment.}

\subsection{Overall Running Time}

Let $\Theta(kr)$ be the sampling rate of our algorithm with character-based sampling.
By \cref{theorem:character-based-sampling}, the maximum number of characters per PE at any time is $\mathcal{O}(\Nsquig/p)$ with  $\Nsquig = N + k^2r\lmax p$.
Since $n \le N/\lmin$ and $N \le n\lmax$, the number of strings per PE then is
$\mathcal{O}(\nsquig/p)$ with $\nsquig = \lmax/\lmin(n + k^2rp)$.

We now combine the running time of the three phases of our algorithm in each level, including $\mathcal{O}(n/p\log{n/p} + D/p)$ time for the initial local string sorting~\cite{DBLP:conf/europar/Ellert0S20}.

\noindent \textbf{1. Distributed Partitioning.} Here, we have to globally sort $\mathcal{O}(rk)$ local sample strings of length $\le \lmax$ per PE.
With RQuick, this is possible in $\mathcal{O}(rk\lmax\log{p}(1 + \beta\log{\sigma}) + \alpha\log^2{p})$ expected time.
Allgathering the $r-1$ splitters needs $\mathcal{O}(\alpha \log{p} + \beta r\lmax \log{\sigma})$ time.
As we have $r\lmax = \mathcal{O}(\Nsquig/p)$ by definition and the local string array as well as the splitters are sorted, computing the local buckets $\mathcal{B}_j^i$ is possible in time $\mathcal{O}(\Nsquig/p)$.%

\noindent \textbf{2. Assignment and Exchange.} The string assignment is dominated by the data exchange which is possible in time $\exchangetilde{p^{\frac{i}{k}}}{\mathcal{O}(\Nsquig/p\log\sigma)}{\mathcal{O}(r)}$ on level $\recursionlevel$.
This holds as a PE stores $\mathcal{O}(\Nsquig/p)$ characters (\cref{theorem:character-based-sampling}) encoded in $\log{\sigma}$ bits each.
By the bounded assignment algorithm each PE exchanges strings with at most $\mathcal{O}(r)$ other PEs (\cref{thm:bounded-assignment}).

\noindent \textbf{3. Local (LCP-aware) Merging.} The $r$ sorted sequences of strings received in the data exchange now needs to be sorted to restore our invariant on the local string array.
While being beneficial for many inputs in practice, asymptotically, LCP-aware merging does not yield substantial advantage.
Thus, we resort to uninformed merging for our analysis.
Processing the $\mathcal{O}(r)$ sorted sequences with $\mathcal{O}(\Nsquig/p)$ characters in total is possible in $\mathcal{O}(\Nsquig/p + \nsquig/p \log{r})$ time, dominating the time required for computing the local buckets during partitioning.

\begin{theorem}%
	\label{thm:multi_level:simple:complexity}
  Multi-level MS with $r=\sqrt[\leftroot{2}k]{p}$, \(\Nsquig = N + k^2r\lmax p\), and \(\nsquig = \lmax/\lmin(n + k^2rp)\), using character-based sampling and
	bounded assignment, runs in expected time\fullstop{}
	\begin{equation*}
		\begin{gathered}
          \mathrm{T_{MS}}(n,N,\lmax) = \mathcal{O}\Biggl(\overbrace{\frac{n}{p}\log{\frac{n}{p}} + \frac{D}{p}}^\text{local sorting} + \overbrace{k\frac{\Nsquig}{p} + \frac{\nsquig}{p}\log{p}}^\text{merging} + \overbrace{k\parens*{\alpha\log^2 p+
					k\sqrt[\leftroot{2}k]{p}\lmax\log p(1 + \beta\log\sigma)}}^\text{partitioning}\Biggr)\\\fullstop{}
			+ \underbrace{k\cdot\exchangetilde{p^{\frac{l}{k}}}{\mathcal{O}\parens*{\frac{\Nsquig}{p}\log\sigma}}{\mathcal{O}\parens*{\sqrt[\leftroot{2}k]{p}}}}_\text{assignment + exchange}\text{.}
		\end{gathered}
	\end{equation*}
\end{theorem}
As we only draw $k\sqrt[k]{p}$ local samples and receive $\sqrt[k]{p}-1$ splitter strings for $k>1$ as opposed to $p$ samples and $p-1$ splitters in the single-level case
we no longer have to compensate the mediocre scalability of the single-level partition phase with a huge amount of data and are capable of sorting small to medium sized inputs on a large number of PEs.
Assuming the exchange primitive $\exchange{P}{h}{\sqrt[k]{p}}$ to run in $\mathcal{O}(\alpha \sqrt[k]{p} + \beta h)$ \cite{axtmann2015practical}, we achieve a latency in $\mathcal{O}(\alpha k\sqrt[k]{p})= o(\alpha p)$ at the cost of an $k$ times higher communication volume.
If we additionally assume $\lmax \le N/(k^2\sqrt[k]{p}p\log{p})$ and $k \le \log{p}/(2\log\log{p})$, we can state a simplified running time of multi-level MS in \cref{corollary:multi-level}.
\begin{corollary}
  \label{corollary:multi-level}
  With the above assumptions, we obtain a running time of multi-level MS in $\mathcal{O}\left(\frac{N}{p}\log{n} + \alpha k\sqrt[k]{p} + \beta k\frac{N}{p}\log{\sigma}\right)$ in expectation.
\end{corollary}

\section{Multi-Level Prefix Doubling Merge Sort}

\label{sec:multi_level_pd_string_sorting}
The distinguishing prefix of $S$ is usually much smaller than total number of characters $N$.
In a distributed algorithm, we can use this property to reduce the communication volume by only exchanging the distinguishing prefixes.
By doing so, instead of explicitly sorting the input strings we obtain the information on where to find the $\maketh{i}$ smallest string of the input.
This, however, is sufficient in many use cases where string sorting is used, e.g., for suffix sorting \cite{DBLP:journals/jacm/KarkkainenSB06}.

Bingmann~et~al. approximate the distinguishing prefix of each string by an upper bound in an iterative doubling process \cite{DBLP:conf/ipps/Bingmann0S20} using a distributed single-shot Bloom filter (dSBF)~\cite{sanders2013communication}.
In each round they hash prefixes with geometrically increasing length of the strings and globally check for uniqueness of the hash values.
If the hash value of a prefix with length $d$ of string $s$ is unique, we find $\distMath{s} \le d$ and $s$ no longer needs to participate in the process.
That way for each string $s$ an approximate distinguishing prefix with length $\distApproxMath{s} \ge \distMath{s}$ can be determined for each string in expected $\mathcal{O}(\log(\dmax))$ rounds assuming a constant false positive probability of the Bloom filter.
This can be achieved with expected latency in $\mathcal{O}(\alpha p \log\dmax)$ and expected bottleneck communication volume in $\mathcal{O}(n/p\log p) + o(D/p\log\sigma)$~\cite{DBLP:conf/ipps/Bingmann0S20}.\footnote{The latency can be reduced by increasing the communication volume by a factor $\Theta(\log{p})$~\cite{DBLP:conf/ipps/Bingmann0S20}.}
By employing a $k$-level Bloom filter for duplicate detection, we generalize this approach to arbitrary levels of indirection.
\begin{theorem}
	\label{thm:k-level-bloomfilter}
	Using communication on a $k$-dimensional
	grid, performing at most $\nmax$ operations (insertions, queries) per PE
	on a dSBF of size
	$m \ge en$ can be done in time
	$\mathcal{O}\left(k\left(\alpha
      p^{1/k}+\beta\nmax\log\frac{mp}{n} + \nmax\log{k}\right)\right)$
	in expectation and with probability $\ge 1 - 1/p^{\omega(1)}$
	assuming the total number of operations $n = \omega(k^2p^{1+1/k}\log{p})$ and additionally $m = \poly(n)$.
\end{theorem}
A proof of \cref{thm:k-level-bloomfilter} can be found in \cref{appendix:distributed-duplicate-detection}.
A problem of using Bloom filters in the prefix doubling process is that the precondition on the overall number of operations required in \cref{thm:k-level-bloomfilter} might not hold when more and more strings drop out of the process because their distinguishing prefixes have already been determined.
We therefore switch to duplicate detection using (atomic) hypercube quicksort for sorting the hash values once there are too few strings left. 
Combining these approaches results in \cref{thm:prefix-doubling}.
\begin{theorem}
  \label{thm:prefix-doubling}
  For each string $s \in \mathcal{S}$ with $\distMath{s} \ge \log{p}/\log{\sigma}$ an approximation $\distApproxMath{s}$ with $\expectedValue{\distApproxMath{s}} = \mathcal{O}(\distMath{s})$ can be computed in time
    \[
      \mathcal{O}\left( \overbrace{\alpha k\sqrt[k]{p}\log{\dmax}}^{\mathrm{latency}} + \overbrace{\beta k \left(\frac{n}{p}\log{p}+ \frac{D}{p}\log{\sigma}\right)}^{\mathrm{communication \ volume}}  + \overbrace{k\frac{n}{p}\log{k}\log\log{\sigma} + \frac{D}{p}}^{\mathrm{local \ work}}\right)
	\]
	in expectation.
	We assume a balanced distribution of strings and their distinguishing prefixes, i.e., $\Theta(n/p)$ strings and $\Theta(D/p)$ per PE, and an overall number of strings $n = \mathcal{O}(\poly(p))$.
    Additionally, we assume $n/p = \omega(k^2\sqrt[k]{p}\log{p}\log\log{p})$ and $k \le \log{p}/(2\log\log{p})$.
\end{theorem}
A proof for \cref{thm:prefix-doubling} can be found in \cref{appendix:section:prefix-doubling}.
Since we assume all strings to be unique, we can also bound the average distinguishing prefix length $D/n = \Omega(\log{n} / \log{\sigma}) = \Omega(\log{p} / \log{\sigma})$.
Hence, there are only few ($\le n/\sigma$) strings with small distinguishing prefixes for which \cref{thm:prefix-doubling} does not yield an (expected) constant factor approximation of the actual distinguishing prefix.
Therefore, we find the expected value of the sum $D_{\approx}$ of the approximate distinguishing prefixes determined with \cref{thm:prefix-doubling} to be in $\mathcal{O}(D).$
In conjunction with the assumed balanced distribution of $D$ over the PEs, we find $\mathcal{O}(\Dsquig/p)$ with $\Dsquig = D + k^2r\dmax p$ as an upper bound on the expected number of characters per PE in the multi-level merge sort algorithm when executed on the approximated distinguishing prefixes only.
Furthermore, since $\dmax = \Omega(\log{n}/\log{\sigma}) = \Omega(\log{p}/\log{\sigma}) $, and we therefore approximate $\dmax$ up to an expected constant factor, we find $\mathrm{T_{MS}}=(n, D, \dmax)$, as an upper bound for the running time of the merge sort part of multi-level \textbf{p}refix \textbf{d}oubling \textbf{m}erge \textbf{s}ort (PDMS) in \cref{thm:multi_level:pdms}.
\begin{theorem}%
	\label{thm:multi_level:pdms}

	Multi-level PDMS with $r=\sqrt[\leftroot{2}k]{p}$, using character-based sampling and
	assignment, and assuming the preconditions of \cref{thm:prefix-doubling}, runs in expected time\fullstop{}
    \begin{equation*}
		\begin{gathered}
          \underbrace{\mathcal{O}\parens*{\overbrace{\log{\dmax}\parens*{\alpha k\sqrt[k]{p}}}^{\textrm{add. latency}} + \overbrace{k\beta\frac{n}{p}\log{p}}^{\textrm{add. comm. volume}} + \overbrace{k\frac{n}{p}\log{k}\log\log{\sigma}}^{\textrm{add. local work}}}}_\text{prefix doubling} + \mathrm{T_{MS}}(n,D,\dmax)\text{.}
		\end{gathered}
	\end{equation*}
\end{theorem}
Note that the $\mathcal{O}(\beta kD/p\log{\sigma} + D/p)$ part of the running time of the prefix doubling process in \cref{thm:prefix-doubling} is subsumed by $\mathrm{T_{MS}}(n,D,\dmax)$ and thus not explicitly stated in \cref{thm:multi_level:pdms}.

\section{Experimental Evaluation}
\label{sec:evaluation}
We now discuss the experimental evaluation of the following distributed-memory algorithms.
\begin{description}[nolistsep,nosep]
	\item[$\ms_k$] Our new multi-level string merge sort with $k$ levels of recursion, see \cref{sec:multi_level_string_sorting}.
	\item[$\pdms_k$] Our new multi-level doubling string merge sort including prefix approximation using grid-wise Bloom filter with $k$ levels of recursion, see \cref{sec:multi_level_pd_string_sorting}.
	\item[$\LcpRQuick$] \RQuick~\cite{DBLP:conf/ipps/Bingmann0S20} with our string-specific optimization, see \cref{sec:algorithmic_building_blocks}.
\end{description}
The single-level variants $\ms_1$, $\pdms_1$ and $\RQuick$ are implementations by Bingmann~et~al.~\cite{DBLP:conf/ipps/Bingmann0S20} that we improved slightly.
We also include the state-of-the-art shared-memory parallel algorithm \psV~\cite{bingmann2013parallel}.
We ran \psV\ on an AMD Epyc Rome 7702P CPU with 64 cores
(\qty{2}{\giga\hertz} base and \qty{3.5}{\giga\hertz} boost frequency) equipped with \qty{1024}{\giga\byte} DDR4 ECC RAM.
All distributed-memory experiments were performed on SuperMUC-NG consisting of \num{6336} nodes (\num{792} nodes per island).
Each node is equipped with two Intel Skylake Xeon Platinum 8174 CPUs (24 cores each, \qty{3.1}{\giga\hertz} base frequency) and \qty{96}{\giga\byte} RAM.
Communication between nodes uses a $\qty[per-mode=single-symbol]{100}{\giga\bit\per\second}$ Omni-Path network.
All algorithms are implemented in C++ and compiled with GCC~11.2.0 with flags \texttt{-O3} and \texttt{-march=native}.
We use the KaMPIng~\cite{KaMPIngArXiv} MPI bindings for our implementations and Open MPI~v4.0.7 for interprocess communication.
Reported times are the average of five runs excluding the first iteration (MPI warm-up phase).

\begin{table}[b]
	\begin{center}
		\caption{Characteristics of real-world data sets used for the strong scaling experiment.}%
		\label{tbl:evaluation:real_world_sets}
		\begin{threeparttable}
			\begin{tabular}{lrrrrrrr}
				\toprule
				                  & \multicolumn{1}{c}{$n$}
				                  & \multicolumn{1}{c}{$N$}
				                  & \multicolumn{1}{c}{$N/n$}
				                  & \multicolumn{1}{c}{$L/n$}
				                  & \multicolumn{1}{c}{$D/N$}
				                  & \multicolumn{1}{c}{$\lmax$}
				\\
				\midrule
				\textsc{CCrawl}   & \qty{2.13}{\giga\nothing}
				                  & \qty{100}{\giga\nothing}
				                  & 46.98
				                  & 31.27
				                  & 0.726
				                  & \qty{1.04}{\mega\nothing}
				\\
				\textsc{Wiki}     & \qty{1.42}{\giga\nothing}
				                  & \qty{97.7}{\giga\nothing}
				                  & 68.59
				                  & 25.82
				                  & 0.415
				                  & \qty{2.07}{\mega\nothing}
				\\
				\textsc{WikiText} & \qty{0.97}{\giga\nothing}
				                  & \qty{81.7}{\giga\nothing}
				                  & 84.21
				                  & 25.27
				                  & 0.336
				                  & \qty{2.07}{\mega\nothing}
				\\
				\bottomrule
			\end{tabular}
		\end{threeparttable}
	\end{center}
\end{table}

All variants use string-based regular sampling with a sampling factor of $2$.
We use \LcpRQuick\ to sort the samples.
For simplicity, we use a grid-wise group (string) assignment in the implementation of our multi-level algorithms.
Here, the PEs are arranged in a $p'\times r$ grid and PEs exchange buckets only along the rows, i.e., PE $i$ sends its bucket for subgroup $j$ along its row to the PE in the $\maketh{j}$ column.
Also, for all-to-all exchanges a simple $k$-dimensional grid all-to-all is used which provides a latency in $\bigO{\alpha \sqrt[k]{p}}$.
LCP compression is used during string exchange phases for data sets where significant common prefixes can be expected.

Multi-level variants always ensure one group per node on the final level
of sorting, i.e., $k$-level variants fall back to $k-1$ or $k-2$ levels for $p<2^{k-1}48$.
For three-level variants, group sizes for the first two levels are chosen such that splitting
factors are as close as possible.%

For strong-scaling experiments, we use real-world data sets, see \cref{tbl:evaluation:real_world_sets} for details.
\begin{description}[nolistsep,nosep]
	\item[\textbf{\textsc{CommonCrawl (CCrawl)}}]
	      consists of the first \qty{100}{\giga\byte} of WET files from the Sep./\allowbreak Oct. 2023 Common Crawl archive
	      (\url{https://index.commoncrawl.org/CC-MAIN-2023-40/}).
	      The WET format consist mostly of plain text with small meta data headers.
	\item[\textbf{\textsc{Wikipedia (Wiki)}}]
	      consists of a dump, from \DTMdate{2023-12-20}, of all pages in the English Wikipedia
	      in XML format \emph{without} edit history
	      (\url{https://dumps.wikimedia.org/enwiki/20231220/}).

	\item[\textbf{\textsc{WikipediaText (WikiText)}}]
	      are from dumps of \textsc{Wikipedia} without any XML metadata.%
\end{description}

\iffigures
	\begin{figure}[t]
		\centering
		\pgfplotsset{
    np ratio legend/.style={
      legend columns=1,
            legend style={at={(-2.5, -0.26)}, anchor=north},
            legend entries={$\ms_1$\\$\pdms_1$\\$\ms_2$\\$\pdms_2$\\$\ms_3$\\$\pdms_3$\\
                    $\RQuick$\\$\LcpRQuick$\\},
        },
}

\begin{tikzpicture}[baseline]
    \begin{groupplot}[
            footnotesize,
            group style={
                    group size=4 by 1,
                    xticklabels at=all,
                    ylabels at=edge left,                    
                    yticklabels at=all,
                    horizontal sep=2em,
                    vertical sep=5em,
                },
            height=.175\textheight, width=.21\textwidth,
            scale only axis, enlarge x limits=0.04,
            enlarge y limits=0.03,
            xtick=data, xmode=log, log basis x=2, xmajorgrids,
            xticklabel=\pgfmathparse{2^\tick}\pgfmathprintnumber{\pgfmathresult},
            x tick label style={rotate=90,/pgf/number format/.cd,fixed,precision=0},
            xlabel={$\text{nodes}\times48/\textnormal{cores}$},
            ylabel={$\text{wall time}/\unit{\second}$},
            label style={font=\footnotesize},
            title style={font=\footnotesize},
            cycle list name=colorListUpToThree,
        ]

        \nextgroupplot[title={$n/p = 10^4$}, ymin=0, ymax=0.6, ytick distance=0.1,
        y tick label style={/pgf/number format/.cd,fixed,fixed zerofill=true,precision=1}]

        \addplot coordinates { (4,0.0505905) (8,0.0662462) (16,0.102683) (32,0.216874) (64,0.449523) (128,1.16063) (256,1.17907) (512,1.95185) (1024,3.79291) };
        \addplot coordinates { (4,0.0418138) (8,0.0557095) (16,0.102141) (32,0.201617) (64,0.433726) (128,1.56009) (256,2.97662) (512,7.05657) (1024,18.6557) };
        \addplot coordinates { (4,0.0560272) (8,0.0587485) (16,0.068184) (32,0.0689761) (64,0.0794039) (128,0.129826) (256,0.174955) (512,0.379632) (1024,0.860719) };
        \addplot coordinates { (4,0.0373176) (8,0.0376604) (16,0.0470978) (32,0.0493042) (64,0.0603443) (128,0.0826471) (256,0.148347) (512,0.298922) (1024,0.66346) };
        \addplot coordinates { (4,0.0897979) (8,0.0910808) (16,0.0919553) (32,0.0936831) (64,0.106478) (128,0.131268) (256,0.163788) (512,0.266722) (1024,0.531366) };
        \addplot coordinates { (4,0.0474277) (8,0.0466564) (16,0.049892) (32,0.0539536) (64,0.0600533) (128,0.0780923) (256,0.109479) (512,0.19862) (1024,0.409279) };

        \addplot coordinates { (4,0.254327) (8,0.291817) (16,0.331052) (32,0.375952) (64,0.444785) (128,0.502336) (256,0.591295) (512,0.662481) (1024,0.86462) };
        \addplot coordinates { (4,0.226327) (8,0.26956) (16,0.300978) (32,0.342181) (64,0.407616) (128,0.459926) (256,0.88805) (512,0.989355) (1024,0.813704) };

        \nextgroupplot[
            title={$n/p = 10^5$},
            ymin=0, ymax=3.2, ytick distance=0.4,
            y tick label style={/pgf/number format/.cd,fixed,fixed zerofill=true,precision=1},
            np ratio legend,
            legend to name={leg:np_ratio_legend}
        ]

        \addplot coordinates { (4,0.644801) (8,0.677757) (16,0.813842) (32,0.832994) (64,1.0635) (128,1.75911) (256,2.88921) (512,5.53908) (1024,12.8824) };
        \addplot coordinates { (4,0.59318) (8,0.606499) (16,0.647423) (32,0.753289) (64,0.98564) (128,1.64705) (256,3.72144) (512,11.1687) (1024,24.9355) };
        \addplot coordinates { (4,0.831777) (8,0.828691) (16,0.839111) (32,0.921098) (64,0.955699) (128,1.01874) (256,1.05302) (512,1.22557) (1024,1.76392) };
        \addplot coordinates { (4,0.633879) (8,0.631919) (16,0.646015) (32,0.662322) (64,0.673849) (128,0.715472) (256,0.788839) (512,0.945914) (1024,1.30923) };
        \addplot coordinates { (4,1.22442) (8,1.20322) (16,1.13623) (32,1.14768) (64,1.19633) (128,1.27094) (256,1.33497) (512,1.50795) (1024,1.87613) };
        \addplot coordinates { (4,0.781282) (8,0.759182) (16,0.730639) (32,0.725067) (64,0.740623) (128,0.761897) (256,0.80277) (512,0.901721) (1024,1.07139) };

        \addplot coordinates { (4,2.88197) (8,3.2947) (16,3.70603) (32,4.14237) (64,4.87946) (128,5.411) (256,6.28294) (512,6.90788) (1024,9.06809) };
        \addplot coordinates { (4,2.64504) (8,3.03981) (16,3.40672) (32,3.78964) (64,4.48999) (128,4.99303) (256,5.79771) (512,6.38561) (1024,8.51642) };

        \nextgroupplot[title={$n/p = 10^6$}, ymin=0, ymax=12, ytick distance=2]

        \addplot coordinates { (4,5.20592) (8,5.43991) (16,5.62522) (32,5.78779) (64,6.57597) (128,8.23696) (256,9.33346) (512,12.3119) (1024,24.1115) };
        \addplot coordinates { (4,4.62864) (8,4.77387) (16,4.83676) (32,5.12225) (64,5.35129) (128,6.57427) (256,8.54862) (512,15.2255) (1024,31.2713) };
        \addplot coordinates { (4,7.17933) (8,7.1421) (16,7.21363) (32,7.37185) (64,7.87024) (128,8.39283) (256,8.3063) (512,8.41634) (1024,11.6387) };
        \addplot coordinates { (4,5.23343) (8,5.23599) (16,5.29987) (32,5.38873) (64,5.41789) (128,5.47008) (256,5.61729) (512,5.84161) (1024,6.97526) };
        \addplot coordinates { (4,11.3501) (8,11.0116) (16,10.3207) (32,10.3167) (64,10.5077) (128,11.4675) (256,11.4457) (512,11.8219) (1024,14.4682) };
        \addplot coordinates { (4,6.65732) (8,6.44196) (16,6.14305) (32,6.1952) (64,6.11161) (128,6.17939) (256,6.32011) (512,6.53512) (1024,7.37365) };

        \nextgroupplot[group/empty plot]
      \end{groupplot}
      \node at ($(group c4r1)-(.75,0)$) {\ref*{leg:np_ratio_legend}};
\end{tikzpicture}
		\caption[%
			Overall times for the weak scaling experiment with variable
			$n/p$ ratio.
		]{%
			Average sorting times (weak scaling) using \textsc{DNData} with
			$\ell=500$ and $D/N=0.5$.
		}
		\label{fig:evaluation:np_ratio}
	\end{figure}
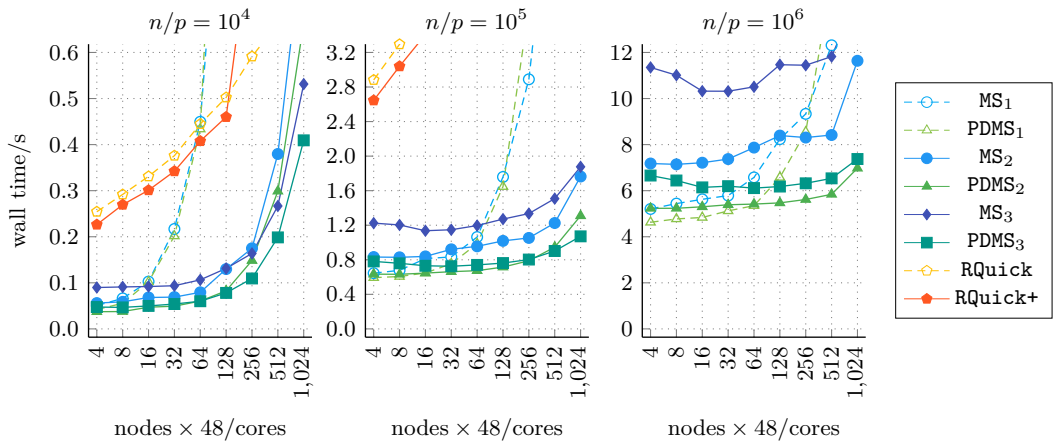
\fi

For our weak-scaling experiments we use the \emph{DNGenerator} string generator \cite{schimek2019distributed, DBLP:conf/ipps/Bingmann0S20} to generate string sets with configurable $D/N$ ratio (denoted by $\textsc{DNData}$).
This allows us to influence running times
of local sorting, the effectiveness of LCP compression, and the length of distinguishing prefixes.

Additional results are given in \cref{appendix:sec:additional_results} for the following experiments.

\subsection{Sorting Small to Medium Sized String Sets}
In this weak-scaling experiment, we evaluate the algorithms on \textsc{DNData} with a string length of $500$ characters, $D/N = 0.5$, and $n/p \in \braces{10^4, 10^5, 10^6}$.
See \cref{fig:evaluation:np_ratio} for the results.
The largest input is near the realistic limit for this system with roughly 2 GB RAM per PE.
We could not run the \RQuick{} variants on it, due to their memory consumption.

The running times broadly confirm the expected relation between input size and scaling behavior of algorithms.
Two-level merge sort significantly outperforms the single-level version on all input sizes for
sufficiently large values of $p$.
For small inputs, adding a third level leads to further improvements from 256 cores on.
As expected, the improvement is most obvious for the smallest inputs with $n/p=10^4$.
Here, the single-level algorithms scale roughly linearly with the number of PEs, as the running times
approximately double for every doubling of $p$.
For $\ms_1$ the scaling behavior can mostly be attributed to the time required for partitioning,
with $\Theta(p)$ samples on each PE needing to be sorted.
For $\pdms_1$ a significant amount of time is also spent in the distinguishing prefix approximation.
The \RQuick{} variants perform significantly worse than our (multi-level) merge sort algorithms.

\subsection{Influence of D/N Ratios}
\iffigures
	\begin{figure*}[t]
		\centering
		\pgfplotsset{
    dn ratio time style/.style={
            ymin=0, ymax=6,
            ytick distance=1
        },
    dn ratio comm style/.style={
            ymin=0, ymax=3.0,
            ytick distance=0.5,
            y tick label style={
                    /pgf/number format/.cd,
                    fixed,fixed zerofill,precision=1
                },
        }
}

\begin{tikzpicture}[baseline]
    \begin{groupplot}[
            footnotesize,
            group style={
                    group size=5 by 1,
                    x descriptions at=edge bottom,
                    y descriptions at=edge left,
                    horizontal sep=.5em,
                    vertical sep=1ex,
                },
            xmajorgrids, scale only axis, enlargelimits=0.05,
            height=.175\textheight, width=.1685\textwidth,
            xlabel={$\text{nodes}\times48/\textnormal{cores}$},
            xmin=4, xmax=128,
            xtick={4, 8, 16, 32, 64, 128}, xmode=log, log basis x=2,
            xticklabel=\pgfmathparse{2^\tick}\pgfmathprintnumber{\pgfmathresult},
            x tick label style={rotate=90,/pgf/number format/.cd,fixed,precision=0},
            label style={font=\footnotesize}, title style={font=\footnotesize},
            cycle list name=colorListUpToTwoWithPSV,
        ]

        \nextgroupplot[
            title={$D/N = 0.0$},
            ylabel={$\text{wall time}/\unit{\second}$},
            dn ratio time style,
        ]

        \addplot coordinates { (4,0.411154) (8,0.470807) (16,0.535027) (32,0.735658) (64,0.883471) (128,1.40988) };
        \addplot coordinates { (4,0.10501) (8,0.127011) (16,0.134759) (32,0.17004) (64,0.225518) (128,0.40711) };

        \addplot coordinates { (4,0.694129) (8,0.713781) (16,0.79631) (32,0.789455) (64,0.860264) (128,0.975862) };
        \addplot coordinates { (4,0.128453) (8,0.132479) (16,0.144491) (32,0.152783) (64,0.164886) (128,0.187977) };
        
        \addplot coordinates { (4,2.10119) (8,2.48125) (16,2.88245) (32,3.39542) (64,3.82968) (128,4.65064) };
        \addplot coordinates { (4,2.01704) (8,2.40037) (16,2.77633) (32,3.29132) (64,3.71572) (128,4.52193) };

        \addplot coordinates { (4,0.117574) (8,0.209896) (16,0.411806) (32,0.810622) (64,1.59666) };

        \nextgroupplot[title={$D/N = 0.25$}, dn ratio time style]

        \addplot coordinates { (4,0.526474) (8,0.569324) (16,0.648955) (32,0.769335) (64,0.969114) (128,1.53645) };
        \addplot coordinates { (4,0.350438) (8,0.374366) (16,0.398711) (32,0.472877) (64,0.619943) (128,1.14218) };

        \addplot coordinates { (4,0.76341) (8,0.785021) (16,0.829133) (32,0.837794) (64,0.91247) (128,0.994602) };
        \addplot coordinates { (4,0.372621) (8,0.383426) (16,0.386335) (32,0.403596) (64,0.417763) (128,0.444609) };
        
        \addplot coordinates { (4,2.45876) (8,2.86151) (16,3.25544) (32,3.76394) (64,4.21841) (128,5.05895) };
        \addplot coordinates { (4,2.31128) (8,2.70483) (16,3.08396) (32,3.55773) (64,3.99077) (128,4.81148) };

        \addplot coordinates { (4,0.567435) (8,0.781142) (16,1.55513) (32,2.86373) (64,5.65101) };

        \nextgroupplot[
            title={$D/N = 0.5$},
            dn ratio time style,
            legend columns=10,
            legend style={at={(0.5,-0.35)},anchor=north,font=\small},
            legend entries={$\ms_1$\\$\pdms_1$\\$\ms_2$\\$\pdms_2$\\\RQuick\\\LcpRQuick\\\psV\\}
        ]

        \addplot coordinates { (4,0.647558) (8,0.702012) (16,0.795281) (32,0.866696) (64,1.05171) (128,1.677) };
        \addplot coordinates { (4,0.585251) (8,0.607831) (16,0.665228) (32,0.78684) (64,1.00089) (128,1.69378) };

        \addplot coordinates { (4,0.825957) (8,0.848163) (16,0.843843) (32,0.940804) (64,0.931586) (128,1.00233) };
        \addplot coordinates { (4,0.634171) (8,0.651667) (16,0.650813) (32,0.653211) (64,0.68046) (128,0.721836) };

        \addplot coordinates { (4,2.88478) (8,3.30099) (16,3.71971) (32,4.23956) (64,4.71348) (128,5.56879) };
        \addplot coordinates { (4,2.63058) (8,3.02017) (16,3.43249) (32,3.90504) (64,4.31173) (128,5.11533) };

        \addplot coordinates { (4,1.00276) (8,1.37545) (16,2.85422) (32,4.86187) (64,9.59154) };

        \nextgroupplot[title={$D/N = 0.75$}, dn ratio time style]

        \addplot coordinates { (4,0.805456) (8,0.841667) (16,0.888476) (32,1.00922) (64,1.24356) (128,1.82874) };
        \addplot coordinates { (4,0.920259) (8,0.980497) (16,1.03966) (32,1.19289) (64,2.17865) (128,2.45614) };

        \addplot coordinates { (4,0.936042) (8,0.937984) (16,0.960939) (32,1.00701) (64,1.0321) (128,1.05578) };
        \addplot coordinates { (4,1.05769) (8,1.06319) (16,1.07639) (32,1.11342) (64,1.12643) (128,1.20645) };
        
        \addplot coordinates { (4,3.36251) (8,3.81432) (16,4.22188) (32,4.75319) (64,5.23164) (128,6.11473) };
        \addplot coordinates { (4,2.99865) (8,3.38219) (16,3.77709) (32,4.2515) (64,4.70226) (128,5.49384) };

        \addplot coordinates { (4,1.45562) (8,1.99027) (16,4.14262) (32,7.02431) (64,13.7173) };

        \nextgroupplot[title={$D/N = 1.0$}, dn ratio time style]

        \addplot coordinates { (4,0.921488) (8,0.930348) (16,0.972435) (32,1.09786) (64,1.3449) (128,1.93506) };
        \addplot coordinates { (4,1.06646) (8,1.08639) (16,1.19635) (32,1.31346) (64,2.08762) (128,2.64836) };

        \addplot coordinates { (4,0.989823) (8,0.989213) (16,0.989678) (32,1.01819) (64,1.01818) (128,1.06388) };
        \addplot coordinates { (4,1.13222) (8,1.14529) (16,1.14494) (32,1.16056) (64,1.15862) (128,1.23573) };
        
        \addplot coordinates { (4,3.77862) (8,4.22839) (16,4.69885) (32,5.22327) (64,5.76321) (128,6.61689) };
        \addplot coordinates { (4,3.2951) (8,3.70894) (16,4.09922) (32,4.56944) (64,5.00077) (128,5.81749) };

        \addplot coordinates { (4,1.97212) (8,2.62969) (16,5.33747) (32,9.00977) (64,17.4628) };

    \end{groupplot}
\end{tikzpicture}
		\caption[%
			Overall sorting times for the weak scaling experiment with variable $D/N$ ratio.
		]{%
			Average sorting times (weak scaling) using \textsc{DNData} with
			$\ell=500$ and $n/p=10^5$.
		}
		\label{fig:evaluation:dn_ratio}
        \vspace*{-.125cm}
	\end{figure*}
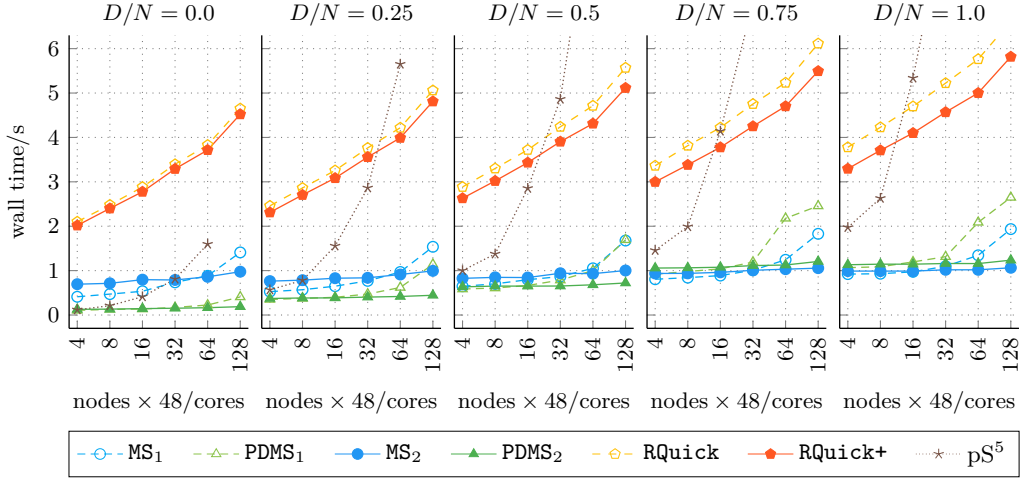
\fi
Our first experiment already shows that multi-level merge sort exhibits improved scaling
properties. %
Now, we have a closer look on the influence of the $D/N$ ratio on the running times of the different algorithms.
As before, we use strings with a length of $500$ characters and $10^5$ strings per PE.
\cref{fig:evaluation:dn_ratio} shows the average running times.
We evaluate $\ms$ and $\pdms$ on one and two levels and the $\RQuick$ variants.
Three-level variants are not part of this experiment as the additional level only yields clearly better performance for $\ge 512$ compute nodes which we did not include in any further experiments due to computing budget constraints.

The influence of $D/N$ ratio is clearly visible for $\ms_k$ and $\pdms_k$, as variants with prefix approximation are superior on instances with ratio up to $0.5$.
For ratios $0.75$ and $1$, the prefix approximation encompasses the whole string and is detrimental to the running time.

LCP compression is highly effective due to the nature of \textsc{DNData} inputs.
As before, multi-level variants outperform single-level counterparts, usually with a crossover
point at \num{32} nodes.
The gap seemingly increases for larger $D/N$ ratios, e.g. for $\ms_k$ on \num{128} nodes the
speedups are \num{1.56} and \num{1.78} for ratios \num{0} and \num{1} respectively.
This can partially be explained if communication volume is considered.
On fewer PEs, two-level variants cause roughly twice the communication volume because string
exchange phases dominate.
At $D/N=0$, sending \qtylist{0.5;1}{\kilo\byte} per string is roughly equivalent to exchanging
every string once or twice respectively.
Communication increases for single-level variants with the number of PEs as partitioning requires
more samples to be sorted.
String exchanges send fewer characters for larger $D/N$ ratios due to LCP compression, sample
sorting remains roughly constant.

The \RQuick{} variants perform worse than the merge sort based algorithms.
However, we can also see that for an increasing $D/N$ ratio improving the algorithms by using LCP-aware merging pays off with \LcpRQuick{} being up to $20 \%$ faster for $D/N = 1$ than plain \RQuick.

The running times of \psV (\emph{shared-memory} using 64 cores) are given for the input of all $p$ processors.
We could not run \psV{} on a node of SuperMUC-NG as these are only equipped with \qty{96} GB RAM.
With $D/N \ge 0.25$, we need $384$ cores to match the performance of \psV{} and only $192$ cores for $D/N \ge 0.5$.
Data with $D/N \approx 0$ is difficult for distributed algorithms since the time spent in local sorting hardly compensates for the for communication overheads and start-up latencies.
Here, we only need \num{1536} cores to be on par with \psV{}.
This is of interest when string sorting is part of more complex distributed tasks as it indicates that sorting the strings directly on the distributed system is faster than transferring them to a sufficiently large shared-memory machine from a very modest number of cores on.

\subsection{Evaluation on Real-World Data}
To evaluate our algorithms on real-world data we use a strong-scaling experiments.
On all three data sets we can observe that the two-level algorithms scale better than their single-level counterparts and finally outperform them.
On the largest PE configuration utilizing \num{256} compute nodes, $\pdms_2$ ist the fastest algorithm on all three data sets.
This highlights the usefulness of our multi-level approaches---even on very skewed real-world inputs.
However, the differences of the running times are less pronounced.
Also, $\pdms_2$ fails on \textsc{CommonCrawl} and $\ms_2$ fails on \textsc{Wikipedia} when using only $4$ compute nodes, due to imbalances of the data.
This shows that string inputs can be inherently hard to partition.

While different sampling and group assignment techniques may improve character balance, the
\textsc{CommonCrawl} data set contains many duplicated strings (e.g., standard
legal dis\-claimers) which are always assigned to a single PE and cannot be shortened by prefix
approximation, resulting in imbalances that we cannot prevent.

\iffigures
	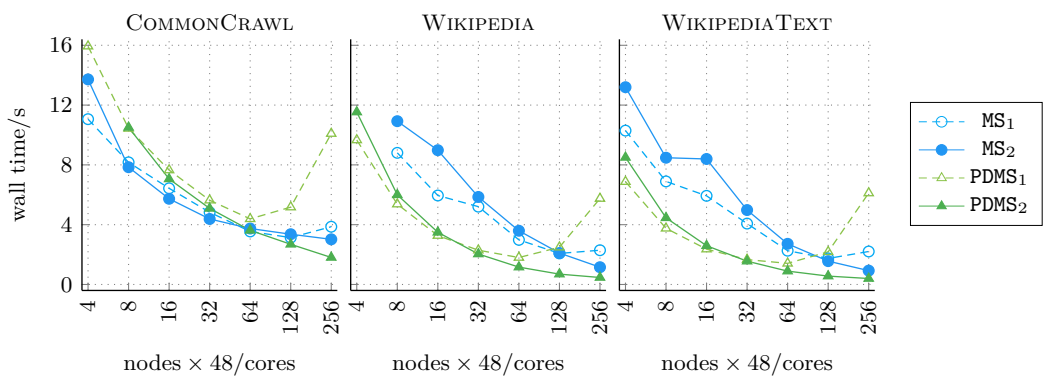
\begin{figure}
		\centering
		\pgfplotsset{
    strong scaling time/.style={
            ymin=0, ymax=16,
            ytick distance=4,
            enlarge y limits=0.025,
            height=.15\textheight,
            width=.24\textwidth,
            y tick label style={/pgf/number format/.cd,fixed,precision=0},
        },
    strong scaling rel/.style={
            ymin=0, ymax=3,
            ytick distance=0.5,
            enlarge y limits=0.035,
            height=.13\textwidth,
            y tick label style={/pgf/number format/.cd,fixed,fixed zerofill,precision=1},
        },
    strong scaling legend/.style={
            legend columns=1,
            legend style={at={(2.45, .75)}, anchor=north},
            legend entries={$\ms_1$, $\ms_2$, $\pdms_1$, $\pdms_2$, $\pdms_2^\triangledown$},
        }
}

\begin{tikzpicture}[baseline]
    \begin{groupplot}[
            footnotesize,
            group style={
                    group size=3 by 1,
                    x descriptions at=edge bottom,
                    y descriptions at=edge left,
                    horizontal sep=0.5em,
                    vertical sep=1ex,
                },
            xmajorgrids, scale only axis,
            enlarge x limits=0.025,
            width=.13\textwidth, xmin=4, xmax=256,
            xlabel={$\text{nodes}\times48/\textnormal{cores}$},
            xtick={4, 8, 16, 32, 64, 128, 256}, xmode=log, log basis x=2,
            xticklabel=\pgfmathparse{2^\tick}\pgfmathprintnumber{\pgfmathresult},
            x tick label style={rotate=90,/pgf/number format/.cd,fixed,precision=0},
            label style={font=\footnotesize}, title style={font=\footnotesize},
            cycle list name=colorListUpToTwoRealWorld,
        ]

        \nextgroupplot[title={\textsc{CommonCrawl}}, ylabel={$\text{wall time}/\unit{\second}$}, strong scaling time]

        \addplot coordinates { (4,11.0586) (8,8.17032) (16,6.43309) (32,4.8306) (64,3.54603) (128,3.134) (256,3.86334) };
        \addplot coordinates { (4,13.7181) (8,7.84432) (16,5.73222) (32,4.38236) (64,3.74246) (128,3.3618) (256,3.02066) };
        \addplot coordinates { (4,15.937) (8,10.4361) (16,7.65281) (32,5.63597) (64,4.38796) (128,5.17303) (256,10.0928) };
        \addplot coordinates { (8,10.5073) (16,7.05837) (32,5.08861) (64,3.62205) (128,2.68828) (256,1.80702) };

        \nextgroupplot[title={\textsc{Wikipedia}}, strong scaling time, strong scaling legend]

        \addplot coordinates { (8,8.80545) (16,5.95178) (32,5.20566) (64,2.98417) (128,2.07758) (256,2.29079) };
        \addplot coordinates { (8,10.9228) (16,8.98173) (32,5.85496) (64,3.59287) (128,2.10247) (256,1.15557) };
        \addplot coordinates { (4,9.6626) (8,5.37327) (16,3.27154) (32,2.29173) (64,1.80233) (128,2.4906) (256,5.7465) };
        \addplot coordinates { (4,11.5339) (8,5.99633) (16,3.49105) (32,2.03873) (64,1.1612) (128,0.692381) (256,0.467425) };

        \nextgroupplot[title={\textsc{WikipediaText}}, strong scaling time]

        \addplot coordinates { (4,10.293) (8,6.89383) (16,5.92974) (32,4.07387) (64,2.25966) (128,1.74101) (256,2.20631) };
        \addplot coordinates { (4,13.1958) (8,8.47844) (16,8.39492) (32,4.97866) (64,2.71963) (128,1.56116) (256,0.926633) };
        \addplot coordinates { (4,6.87536) (8,3.75643) (16,2.36202) (32,1.66499) (64,1.41277) (128,2.18468) (256,6.11441) };
        \addplot coordinates { (4,8.49946) (8,4.45971) (16,2.59265) (32,1.55821) (64,0.893607) (128,0.560526) (256,0.398876) };
    \end{groupplot}
\end{tikzpicture}
		\caption{Average running times (strong scaling) using real-world inputs (see \cref{tbl:evaluation:real_world_sets}).}
		\label{fig:evaluation:strong_scaling}
        \vspace*{-.225cm}
	\end{figure}
\fi

\section{Conclusion And Future Work}
We demonstrate---in theory and practice---that string sorting can be scaled to very large number of processors.
Our best algorithm, a multi-level prefix-doubling merge sort, only requires internal work and communication volume close to the optimum (the total length of all distinguishing prefixes) per level.
In practice, all our multi-level algorithms outperform their single-level counterparts on a wide range of inputs on up to \num{49152} cores (from a modest number of cores on).
This is especially of importance in scenarios where string sorting is part of a distributed application, i.e., it is not feasible to sort the data on a large shared-memory machine because of the transfer costs.
Hence, we see our work as an important building block to enable more complex string-processing tasks  at a massively parallel scale.

One problem we want to tackle in the future is suffix sorting based on sorting strings of equal length~\cite{DBLP:journals/jacm/KarkkainenSB06}.
To this end, we plan to extend our algorithms to support \emph{space-efficient} string sorting.
Some inputs are highly compressible with lots of overlapping strings, e.g., the suffixes of a text with length $n$ with a combined length of approximately $n^2/2$ can be represented using only $n$ characters.
However, due to the scarcity of main memory on most supercomputers, we cannot easily materialize \emph{all} strings at the same time during sorting.

\bibliography{short}

\clearpage

\section{Distributed Partitioning}
Due to the multidimensionality of the string sorting problem, determining balanced partitions is far more challenging than in atomic sorting.
Some of the steps of our merge sort algorithm depend on the number of strings in the local string array, while others depend on the number of characters or the size of the distinguishing prefix.
We therefore adapt \emph{string-based} and \emph{character-based} partitioning~\cite{DBLP:conf/ipps/Bingmann0S20} to our multi-level approach.
These schemes bound the number of strings and characters, respectively.
The general approach is to draw and globally sort a number of samples on each PE.
Then, $r-1$ splitters $f_j$ are chosen from the global sample array.
Since all strings are locally sorted before the partitioning step, we can make use of a \emph{regular} sampling approach \cite{DBLP:journals/pc/LiLSS93, DBLP:journals/jpdc/ShiS92} in which the samples are drawn equidistantly.
\subsection{String Based Partitioning}
\label{appendix:sec:string_based_partitioning}
On the $\maketh{\recursionlevel}$ recursion level, there are $r^{\recursionlevel-1}$ groups of PEs working on independent sorting problems, see \cref{fig:merge-sort-scheme}.
We now describe the partitioning process from the point of view of one such group.
Let $\StringSet{S'}$ be the concatenation of the local string arrays of the PEs of one such group of size $p'$.
Unlike for single-level MS, now, we cannot assume an equal number of strings on each PE as from the second level of recursion on, this number is subject to the result of a previous partitioning round which is not exact.
For multi-level MS, string-based partitioning consists of the following two steps:\\

\begin{description}
	\item[Local Sampling:] Let $v > 0$ be the \emph{sampling} factor.
	      In total, there will be $p'(v+1)$ samples drawn from the local string arrays.
	      To simplify the discussion, we assume $\abs{\StringSet{S'}}$ to be divisible by $p'(v+1)$.
          This can be lifted when allowing for a small additional imbalance factor $(1 + 1/k)$ (see \cref{appendix:partitioning-general-case}).
	      Let $\omega = \abs{\StringSet{S'}}/(p'(v+1))$.
	      PE $i$ then draws $\ceil{\abs{\StringSetIndex{S}{i}}/\omega} - 1$ samples $\mathcal{V}_i$ from its local string array spaced as evenly as possible.
	      $\mathcal{V}$ is the union of all local samples.
	      If $\abs{\mathcal{V}} < p'(v+1)$, then the first $p'(v+1) - \abs{\mathcal{V}}$ PEs draw one additional sample.
	      This ensures that at most \(\omega\) strings are between two local samples on each PE.
	      Also, the global number of samples is a multiple of $p'$ and of $r$ as we find $p' = r^{k+1-\recursionlevel}$ on recursion level $\recursionlevel$.

	\item[Splitter Computation:] The samples $\mathcal{V}$ are globally sorted using hypercube quicksort.
	      Then $r-1$ splitters $f_j = \mathcal{V}[j|\mathcal{V}|/r -1]$ for $0 < j < r$ are determined using a prefix-sum.
	      Subsequently, the $r-1$ splitters are communicated to all PEs using an all-gather operation.
\end{description}

Although we do not necessarily draw the same number of samples from each PE, we can apply the adapted \emph{sample-density} lemma:
\begin{lemma}[{\cite[Lemma 1.1]{DBLP:conf/ipps/Bingmann0S20}}]
	\label{appendix:lemma:sample-density-string}
	For $i = 0, \ldots, (p'-1)$ let $\StringSet{S}''_i = \{s \in \StringSet{S}_i' \mid a \le s \le b \}$ be an arbitrary contiguous subarray of $\StringSet{S}_i'$.
	If $|\StringSetIndex{S}{i}'' \cap \StringSetIndex{V}{i}| = k$, then $|\StringSetIndex{S}{i}''| \le (k+1)\omega$ with $\omega = |\StringSet{S}'| / (r(v+1))$.
\end{lemma}
This gives us bounds for the number of strings per bucket.
\begin{theorem}
	\label{appendix:theorem:string-sampling-r-streams}
	Using string-based partitioning with a sampling factor of $v$, every bucket $\Bucket^j$ contains at most $\left (1+\frac{r}{v} \right ) \frac{\abs{\StringSet{S}'}}{r}$ strings.
\end{theorem}
\begin{proof}
    The proof transfers from the one of a similar claim \cite[Theorem 2]{DBLP:conf/ipps/Bingmann0S20}.
	By counting the number of sample strings that are contained within a bucket $\StringSet{B}^j$ and applying \cref{appendix:lemma:sample-density-string}, we obtain the claimed bounds.
    Let $\mathcal{V}_i^j=\mathcal{B}_i^j\cap\mathcal{V}_i$ be the samples on PE $i$ that fall into the
	$\maketh{j}$ bucket.
	Using \cref{appendix:lemma:sample-density-string} it follows that
	$\mathcal{B}_i^j\leq\parens{|\mathcal{V}_i^j|+1}\omega$.
	By definition in step \ref{alg:multi_level:simple:partition} of the partition algorithm, splitters $f_j$ and $f_{j+1}$ are separated by
	at most $p'(v+1)/r-1$ elements and thus $\sum_{i=0}^{p'-1} \abs{\mathcal{V}_i^j} = p'(v+1)/r$ when including
	$f_{j+1}$.
	Using the definition of $\mathcal{B}^j$ we find \fullstop{}

	\begin{align*}
		\abs{\mathcal{B}^j} & = \sum_{i=0}^{p'-1} \abs{\mathcal{B}_i^j}
		\leq \sum_{i=0}^{p'-1} \parens*{\abs{\mathcal{V}_i^j}+1}\omega                            \\
		                    & \leq \omega \parens*{\frac{p'(v+1)}{r} + p'}
		= \frac{\abs{\StringSet{S}'}}{p'(v+1)}\parens*{\frac{p'(v+1)}{r} + p'}                    \\
		                    & = \frac{\abs{\StringSet{S}'}}{r} + \frac{\abs{\StringSet{S}'}}{v+1}
		\leq \frac{\abs{\StringSet{S}'}}{r} + \frac{\abs{\StringSet{S}'}}{v} = \left (1+\frac{r}{v} \right ) \frac{\abs{\StringSet{S}'}}{r}
	\end{align*}
	\par \vspace{-1.7\baselineskip}
	\qedhere
\end{proof}

To summarize, using string-based regular sampling yields buckets with a maximum imbalance of $n/v$ strings.
Choosing $v = \Theta(r)$ implies a number of strings per bucket in $\Theta(\mathcal{S}'/r)$.
\begin{lemma}[\cref{lemma:string-based-sampling}]
	On recursion level $\recursionlevel$ with $r = \sqrt[k]{p}$ in step 2) of multi-level MS, string-based regular sampling with sampling factor $v$ yields a maximum bucket size of
	\[
		|\StringSet{B}^j| \le \left(1 + \frac{r}{v} \right)^\recursionlevel \frac{n}{r^\recursionlevel}.
	\]
\end{lemma}
\begin{proof}
	We give a proof by induction. For the first level of recursion, \cref{appendix:theorem:string-sampling-r-streams} directly yields a maximum bucket size of
	\[
		|\StringSet{B}^j| \le \left(1 + \frac{r}{v} \right) \frac{n}{r}.
	\]
	Applying the bounds for recursion level $\recursionlevel-1 \rightsquigarrow \recursionlevel$ to the same theorem  results in
	\[
		\abs{\StringSet{B}^j} \le \left(1 + \frac{r}{v} \right) \frac{\left(1 + \frac{r}{v} \right)^{\recursionlevel-1} \frac{n}{r^{\recursionlevel-1}}}{r}  = \left(1 + \frac{r}{v} \right)^\recursionlevel \frac{n}{r^\recursionlevel}.
	\]
\end{proof}

The term $(1+r/v)^{\recursionlevel}$ signifies that the imbalance between the buckets  multiplies with each level of recursion.
Therefore, we need to choose $v=\Theta(kr)$ for any number of \(k\) levels to keep the term asymptotically constant during the entire sorting process.
For single-level MS, $v = \Theta(r) = \Theta(p)$ samples per PE are sufficient. %
Note the large difference between drawing $kr$ samples per PE (on average) here and $r^k = p$ samples in the single-level case.
Using the assignment strategy described in \cref{section:assignment-strategies}, which equally distributes the strings over the PEs in each group, we arrive at \cref{thm:string-based-sampling}.

\begin{theorem}[String-Based Sampling, \cref{thm:string-based-sampling}]
	\label{appendix:thm:string-based-sampling}
	Using a sampling factor of $v = \Theta(kr) = \Theta(k\sqrt[k]{p})$ the number of strings per PE is in $\bigO{n/p}$ in each level of the algorithm.
\end{theorem}

\subsection{Partitioning Without Loss of Generality}
\label{appendix:partitioning-general-case}

Let $\omega = \abs{\mathcal{S'}}/(p'(v+1))$ and $\omega$ is not an integer, i.e., $\lceil \omega \rceil = 1 + \lfloor \omega \rfloor$.
We then use $p'(v+1)$ (potentially overlapping) samples of size $\lceil \omega \rceil$.

\begin{theorem}
	Using string-based partitioning with a sampling factor of $v$ and $\abs{S'} = \Omega(p'(v+1)k)$, every bucket $\Bucket^j$ contains at most $\left (1+\frac{r}{v} \right ) \frac{\abs{\StringSet{S}'}}{r}(1 + \frac{1}{k})$ strings.
\end{theorem}
\begin{proof}
	Let $\mathcal{V}_i^j=\mathcal{B}_i^j\cap\mathcal{V}_i$ be the samples on PE $i$ that fall into the
	$\maketh{j}$ bucket.
	Using \cref{appendix:lemma:sample-density-string} it follows that
	$\abs{\mathcal{B}_i^j}\leq\parens{|\mathcal{V}_i^j|+1}\lceil \omega \rceil$.
	By definition in step \ref{alg:multi_level:simple:partition} of the partition algorithm, splitters $f_j$ and $f_{j+1}$ are separated by
	at most $p'(v+1)/r-1$ elements and thus $\sum_{i=0}^{p'-1} \abs{\mathcal{V}_i^j} = p'(v+1)/r$ when including
	$f_{j+1}$.
	Using the definition of $\mathcal{B}^j$ we find \fullstop{}

	\begin{align*}
		\abs{\mathcal{B}^j} & = \sum_{i=0}^{p'-1} \abs{\mathcal{B}_i^j}
		\leq \sum_{i=0}^{p'-1} \parens*{\abs{\mathcal{V}_i^j}+1}\lceil \omega \rceil                                                                     \\
		                    & \leq \lceil \omega \rceil \parens*{\frac{p'(v+1)}{r} + p'} = (\lfloor \omega \rfloor + 1) \parens*{\frac{p'(v+1)}{r} + p'} \\
		                    & \le \omega \parens*{\frac{p'(v+1)}{r} + p'} +\parens*{\frac{p'(v+1)}{r} + p'}                                              \\
		                    & \le \omega \parens*{\frac{p'(v+1)}{r} + p'} +\frac{\omega}{k}\parens*{\frac{p'(v+1)}{r} + p'}                              \\
		                    & = \omega \parens*{\frac{p'(v+1)}{r} + p'}\parens*{1 + \frac{1}{k}}                                                         \\
		                    & = \frac{\abs{\StringSet{S}'}}{p'(v+1)}\parens*{\frac{p'(v+1)}{r} + p'} \parens*{1 + \frac{1}{k}}                           \\
		                    & = \left(\frac{\abs{\StringSet{S}'}}{r} + \frac{\abs{\StringSet{S}'}}{v+1}\right)\parens*{1 + \frac{1}{k}}
		\leq \parens*{\frac{\abs{\StringSet{S}'}}{r} + \frac{\abs{\StringSet{S}'}}{v}} \parens*{1 + \frac{1}{k}}                                         \\
		                    & = \left (1+\frac{r}{v} \right )\parens*{1 + \frac{1}{k}} \frac{\abs{\StringSet{S}'}}{r}
	\end{align*}
	\par \vspace{-1.7\baselineskip}
	\qedhere
\end{proof}

\subsection{Character Based Sampling}
\label{appendix:sec:character_based_partitioning}
We generalize \emph{character}-based regular sampling \cite{DBLP:conf/ipps/Bingmann0S20} to our multi-level approach to achieve tighter bounds on the number of characters per PE than the conservative $\mathcal{O}(\lmax n/p)$.
Now, each PE of the considered group draws $\ceil{\norm{\StringSet{S'}_i}/ \omega'} - 1$ equally spaced samples from its character array with sampling distance $\omega'= \norm{\StringSet{S'}}/(p'(v+1))$.
To arrive at the final string samples, we shift the sampled character positions by at most $\lmax - 1$ characters to beginning of the string.
If the total number of samples is smaller than $p'(v+1)$, the first PEs draw one additional sample.

Bingmann~et~al. proposed to shift the sampled character positions by at most $\lmax -1$ positions to the beginning of the next string.
Furthermore, they required $\ell < \omega'$.
We will show that this restriction is not necessary and therefore slightly adapt their proof of {\cite[Lemma~2.2]{DBLP:conf/ipps/Bingmann0S20}}.
Note that our approach might result in sampling the same string multiple times.
In this case, the string must be unified by concatenating its (local) index which requires at most $\mathcal{O}(\log(n/p)/\log{\sigma})$ additional characters.
As we assume all strings to be unique, this is in $\mathcal{O}(\lmax)$.

Again, we state a slightly adapted version of the \emph{sample density} lemma for character-based sampling:
\begin{lemma}
	[{\cite[Lemma~2.2]{DBLP:conf/ipps/Bingmann0S20}}]%
	\label{lemma:sample-density-character}
	Let $\mathcal{S}_i''=\braces{s\in\mathcal{S}_i' \mid a\leq s\leq b}$ be an arbitrary contiguous
	subsequence of $\mathcal{S}_i'$ for $i\in\braces{0, \dots, p'-1}$.
	With $\abs{\mathcal{S}_i'\cap\mathcal{V}_i}=k$ it must hold that
	$\norm{\mathcal{S}_i''} \leq(k+1)\parens{\omega'+\lmax}$.
\end{lemma}
\begin{proof}
  Initially, there have been at most $\omega' - 1$ characters between two sampled positions.
  Since each sampled character position has been moved at most $\lmax$ characters, there are at most $\omega' -1 + \lmax$ characters between two shifted sampled strings.
  If $k=0$, then all elements of $\mathcal{S}_i''$ are either contained between two consecutive samples of $\mathcal{V}_i$ or smaller than the first sample or greater than the last.  Thus, $\norm{\mathcal{S}_i''} \le \omega' - 1 + \lmax$ (or $\norm{\mathcal{S}_i''} \le \omega' + \lmax$ in the latter case).
  For $k=1$, let $c$ be the first character of the sample string contained in $\mathcal{S}_i''$.
  We can split the character array $\mathcal{C}(\mathcal{S}_i'')$ of $\mathcal{S}_i''$ into two character arrays $\mathcal{C}_< +  c + \mathcal{C}_>$.
  For each of the two character arrays, $k=0$ applies and we thus find $\norm{\mathcal{S}_i''} \le (\omega' -1 + \lmax) + 1 + (\omega' + \lmax) \le 2(\omega' + \lmax)$.
  For $k \ge 2$, we can split the corresponding character array of $\mathcal{S}_i''$ into $(k+1)$ subarrays of which the first contains at most $\omega' -1 + \lmax$ characters,
  the last one at most $\omega' + \lmax$ characters and all in between at most $\omega' - 1 + \lmax$ characters.
  Adding the first character of the $k$ sample strings then yields $\norm{\mathcal{S}''_i} \le (k+1)(\omega' + \lmax)$.
\end{proof}

We now arrive at the following bounds for character-based partitioning into $r$ buckets.

\begin{theorem}
	\label{thm:techniques:partition:regular:character:bucket_size}
	Using character-based regular sampling with a sampling factor of $v$, every bucket $\mathcal{B}^j$
	obtained using the partition algorithm contains at most $\parens*{1+\frac{r}{v}} \frac{\norm{\mathcal{S'}}}{r}+ \parens*{1 + \frac{v+1}{r}}p'\lmax$ characters.
\end{theorem}
\begin{proof}
	As for its string-based counterpart, by counting the number of sample strings that are contained within a bucket $\StringSet{B}^j$ and applying \cref{lemma:sample-density-character}, we obtain the claimed bounds on the maximum number of characters per bucket.
    Let $\mathcal{V}_i^j=\mathcal{B}_i^j\cap\mathcal{V}_i$ be the local sample buckets as in the proof
	of Theorem~\ref{appendix:theorem:string-sampling-r-streams}.
	Using Lemma~\ref{lemma:sample-density-character} yields the analogous bound
	$\abs{\mathcal{B}_i^j}\leq\parens*{\abs{\mathcal{V}_i^j}+1}\parens{\omega'+\lmax}$.
	The equalities for $\mathcal{V}_i^j$ and $\mathcal{B}^j$ remain identical which suffices to obtain
	the stated bound.
	\begin{align*}
		\norm{\mathcal{B}^j} & = \sum_{i=0}^{p'-1} \norm{\mathcal{B}_i^j}
		\leq \sum_{i=0}^{p'-1} \parens*{\abs{\mathcal{V}_i^j}+1}\parens{\omega'+\lmax}                                            \\
		                     & \le \parens*{\frac{p'(v+1)}{r}+p'}\parens*{\frac{\norm{\mathcal{S'}}}{p'(v+1)}+\lmax}              \\
		                     & \le \frac{\norm{\mathcal{S'}}}{r}+\frac{\norm{\mathcal{S'}}}{v}+\parens*{\frac{p'(r+v+1)}{r}}\lmax \\
		                     & = \parens*{1+\frac{r}{v}} \frac{\norm{\mathcal{S'}}}{r}+ \parens*{1 + \frac{v+1}{r}}p'\lmax
	\end{align*}
	\qedhere
\end{proof}

Note that the term $\parens{1 + \frac{v+1}{r}}p'\lmax$ in
Theorem~\ref{thm:techniques:partition:regular:character:bucket_size} can be simplified to
$\Theta(p'\lmax)$ if $v=\Theta(r)$.
Hence, the bound on the number of characters in a bucket $\mathcal{B}^j$ does not only depend on the number of samples but also on the number of PEs as a (character-)shift of up to $\lmax - 1$ might occur on every PE.
In \cref{lemma:multi_level:complexity:character_bucket_size}, we give bounds on the number of characters per bucket over the course of our algorithm when using character-based partitioning.
\begin{lemma}[\cref{lemma:multi_level:complexity:character_bucket_size}]
	\label{appendix:lemma:multi_level:complexity:character_bucket_size}
	On recursion level $\recursionlevel$ with $r=\sqrt[\leftroot{2}k]{p}$ in
	step~\ref{alg:multi_level:simple:partition} of multi-level MS using character-based regular
	sampling with a sampling factor of $v$, each bucket contains at most \(\norm{\mathcal{B}^j}\leq \parens*{1+\frac{r}{v}}^{\recursionlevel}\parens*{\frac{N}{r^\recursionlevel}	+ \recursionlevel \parens*{1+\frac{v+1}{r}}\frac{p}{r^{\recursionlevel-1}}\lmax}\) characters.
\end{lemma}
\begin{proof}
	We give a proof by induction.
	The base case of $\recursionlevel=1$ can be derived immediately from
	Theorem~\ref{thm:techniques:partition:regular:character:bucket_size} as follows:\fullstop{}
	\begin{equation*}
		\norm{\mathcal{B}^j}
		\leq \parens*{1 + \frac{r}{v}}\frac{N}{r}+ \parens*{1+\frac{v+1}{r}}p\lmax
		\leq \parens*{1 + \frac{r}{v}}\parens*{\frac{N}{r}+ \parens*{1+\frac{v+1}{r}}p\lmax}
	\end{equation*}
    as $r/v\geq 0$ and therefore $(1 + r/v) \ge 1$ holds.
	We now proceed with the inductive case $\recursionlevel-1\leadsto \recursionlevel$ by applying the same theorem again on each
	group.
	Note that the value of $p'$ needs to be changed to reflect the current group size, i.e.,~
	$p/r^{\recursionlevel-1}$.
	With that we find:\fullstop{}
	\begin{align*}
         & \norm{\mathcal{B}^j} \le \parens*{1+\frac{r}{v}}\frac{\mathrm{Induction \ Hypothesis \ for \ }(\recursionlevel - 1)}{r} + \parens*{1+\frac{v+1}{r}}\frac{p}{r^{\recursionlevel-1}}\lmax            \\
         &= \parens*{1+\frac{r}{v}}\frac{\parens*{1+\frac{r}{v}}^{\recursionlevel-1}\parens*{\frac{N}{r^{\recursionlevel-1}} + (\recursionlevel-1)\parens*{1+\frac{v+1}{r}}\frac{p}{r^{\recursionlevel-2}}\lmax}}{r} + \parens*{1+\frac{v+1}{r}}\frac{p}{r^{\recursionlevel-1}}\lmax \\
         & = \parens*{1+\frac{r}{v}}^\recursionlevel\parens*{\frac{N}{r^\recursionlevel} + (\recursionlevel-1) \parens*{1+\frac{v+1}{r}}\frac{p}{r^{\recursionlevel-1}}\lmax}  +  \parens*{1+\frac{v+1}{r}}\frac{p}{r^{\recursionlevel-1}}\lmax                                                                            \\
         & \leq \parens*{1+\frac{r}{v}}^\recursionlevel\parens*{\frac{N}{r^\recursionlevel} + (\recursionlevel-1) \parens*{1+\frac{v+1}{r}}\frac{p}{r^{\recursionlevel-1}}\lmax} +  \underbrace{\parens*{{1+\frac{r}{v}}}^{\recursionlevel}}_{\ge 1}\parens*{1+\frac{v+1}{r}}\frac{p}{r^{\recursionlevel-1}}\lmax                                                                            \\
         & = \parens*{1+\frac{r}{v}}^\recursionlevel\parens*{\frac{N}{r^\recursionlevel} + \recursionlevel\parens*{1+\frac{v+1}{r}}\frac{p}{r^{\recursionlevel-1}}\lmax} \\
	\end{align*}
\end{proof}

The additional term depending on $\lmax$ stems from the shifting to the beginning the of the strings,
By distributing the characters equally over the PEs in each group up to additional $\mathcal{O}(\lmax)$ characters (see \cref{section:assignment-strategies}), we can limit the maximum number of characters per PE in \cref{appendix:theorem:character-based-sampling}.%
\begin{theorem}[Character-Based Sampling, \cref{theorem:character-based-sampling}]
	\label{appendix:theorem:character-based-sampling}
	Using a sampling factor in $\Theta(kr)$ the maximum number of character per PE is in
	\[
		\mathcal{O}\parens*{\frac{N}{p} + k^2r\lmax}.
	\]
\end{theorem}

For the single-level case with $k=1$, this is equivalent to $O(N/p + p\lmax)$ which is the bound in the original algorithm \cite{DBLP:conf/ipps/Bingmann0S20}.
For $k > 1$, we even have an improvement over the single-level algorithm.
Since we assume $k$ in $\bigO{\log p/\log\log{p}}$, we find $k^2r\lmax = \mathcal{O}(\log^2(p)\sqrt[k]{p}\lmax) = o(p\lmax)$.
This may seem counter-intuitive at first as we introduce a potential imbalance already in the first recursion level.
However, the subsequent assignment step distributes this imbalance equally over $p'/r$ PEs.

\section{Distributed Duplicate Detection}
\label{appendix:distributed-duplicate-detection}
A key building block of the computation of the approximate distinguishing prefixes approach is distributed duplicate detection.
Bingmann et al.~\cite{DBLP:conf/ipps/Bingmann0S20} use a distributed single-shot Bloom filter \cite{sanders2013communication} for this task, which allows for approximate membership queries.
Conceptually, a distributed single-shot Bloom filter of size $m$ is a bit array of size $m$ which is equally distributed over the $p$ PEs.
To insert an element $e$, the element is hashed to a random position $h(e)$ within the interval $[0,m)$ and the corresponding bit is set on the PE responsible for this position, where $h$ denotes a random hash function.
For querying whether an element $e$ is contained in the filter, one simply has to check whether the corresponding bit at position $h(e)$ is set.
This may result in a false positive result.
A false positive query probability $f^+$ for a single-shot Bloom filter with $n$ inserted elements can be achieved by having a size $m \ge nf^+$ \cite{sanders2013communication}.
In our setting, we assume that operations are executed in batches, i.e., each PE $i$ has $n_i$ insertion or query operations.
The hash values associated with the operations are compressed before sending, for example using Elias-Fano encoding \cite{Elias74,Fano71}.
It can be shown that the expected communication volume for each insertion/query operation is in $\log{p} + \mathcal{O}(1)$ bits for Bloom filters with constant false positive probability.
For more details, we refer to \cite{sanders2013communication}.

A problem of distributed (single-level) single-shot Bloom filters is that they have a latency (at least) linear in $p$ since messages are delivered directly \footnote{Note that to mitigate this, Bingmann~et~al. propose hypercube all-to-all communication, which reduces the latency to $\mathcal{O}(\alpha \log{p})$ at the cost of a communication volume increased by a logarithmic factor.}.
As another restriction, we find that the expected communication volume in $\mathcal{O}(\log{p})$ per operation is only guaranteed if the overall number of queries is large.

To mitigate both problems, we propose to apply our $k$-level communication scheme also to Bloom filters, which reduces their latency to $\mathcal{O}(\alpha k\sqrt[k]{p})$ at the cost of a communication volume of $\mathcal{O}(k\log{p})$ bits per operation.
Furthermore, in \cref{appendix:k-level-bloomfilter} we will show that $k$-level Bloom filters only require $\omega(k^2p^{1+1/k}\log{p})$ operations for an expected communication volume in $\mathcal{O}(k\log{p})$.

However, since the number of strings participating in the prefix doubling process decreases from iteration to iteration, we might eventually end up with less queries from some iteration on.
When this is the case, we will use (atomic) hypercube quicksort for duplicate detection.
An upper bound for the running time for this approach will be given in \cref{appendix:thm:qsort_duplicate}.
\subsection{Multi-Level Bloom filters}

Before discussing the running time of multi-level distributed single-shot Bloom filters in \cref{appendix:k-level-bloomfilter}, we will prove \cref{appendix:lemma:balls-into-bins-with-duplicates}, which analyzes a balls-into-bins with duplicates scenario.
\cref{appendix:lemma:balls-into-bins-with-duplicates} will then be used as a building block in the proof of \cref{appendix:k-level-bloomfilter}.

Let $M$ be a multiset consisting of $m$ elements on a group of $p'$ processors, where each PE holds $\mbar = m/p'$ elements.
Furthermore, we assume the elements to be locally unique.
Let $S$ denote the set of unique elements within $M$ and for $e \in S$ let $c_e$ denote the number of occurrences of $e$ within $M$.
Each local element is now assigned uniformly at random to one of the $p'$ PEs.
\begin{lemma}
	\label{appendix:lemma:balls-into-bins-with-duplicates}
	Let $X$ denote the maximum number of elements a PE receives in the above-described process.
	For $k \ge 1$ and $\mbar = \omega(k^2p'\log{p})$ with $p \ge p'$, we find
	\[
		\mathrm{Pr}[X > \mbar(1+1/k)] \le \frac{1}{p^{\omega(1)}}.
	\]
\end{lemma}
\begin{proof}(Outline)
	In our proof outline, we mainly reiterate the idea behind the proof of \cite[Lemma 3]{sanders2013communication}.
	For each (unique) element $e \in S$, let $\mathbbm{1}_{e,j}$ be an indicator variable with $\mathbbm{1}_{e,j} = 1$ iff element $e$ is sent to PE $j$.
	As the elements are assigned uniformly at random, we have $\mathrm{Pr}[\mathbbm{1}_{e,j} = 1] = 1 /p' $.
	Let $X_j$ denote the number of elements received by PE $j$.
	We find $X_j = \sum_{e\in S}c_e\mathbbm{1}_{e,j}$.
    Since $\sum_{e \in S} c_e = m$, we find $\expectedValue{X_j} = m/p'$ for the expected number of elements received by PE $j$ by the linearity of expectation.
	With the same arguments as in \cite{sanders1996competitive} one can show that $X_j$ is least sharply concentrated around its mean when all elements have a multiplicity $c_e = p'$,
	i.e., each PE has the same set of elements.
	It is therefore sufficient to show a bound on $\mathrm{Pr}[Y_j > \mbar/p'(1+1/k)]$ for this worst-case setting in which we have $X_j = p'Y_j$, where $Y_j$ denotes the number of unique elements received by PE $j$.

	As $Y_j \sim\mathrm{Binom}(\mbar,1/p')$, we can use a Chernoff-bound \cite[Theorem 4.4]{mitzenmacher05} to obtain
	\[
		\textrm{Pr}\left[Y_j > \frac{\mbar}{p'}\parens*{1 + \frac{1}{k}}\right] \le \exp\parens*{-\frac{\mbar}{3p'k^2}}
	\]
	as an upper bound on the probability that PE $j$ receives more than $\mbar/p'(1+1/k)$ (unique) elements for $k \ge 1$.
	By the union bound argument, the probability that any of the $p'$ PE receives more than $\mbar/p(1+1/k)$ (unique) elements is at most
	\begin{equation*}
		p'\exp\left(-\frac{\mbar}{3p'k^2}\right) \le p'\exp(-\omega(\log{p'}))  = \frac{p'}{p^{\omega(1)}} = \frac{1}{p^{\omega(1)}}
	\end{equation*}
	for $\mbar = \omega(p'k^2\log{p})$. Since $X_j = p'Y_j$ for all PE $j$, we overall find $\mathrm{Pr}[X > \mbar(1 + 1/k)] \le \frac{1}{p^{\omega(1)}}$.
\end{proof}

We now use \cref{appendix:lemma:balls-into-bins-with-duplicates} as a building block to show a bound on the running time for a multi-level Bloom filter.

\begin{theorem}[\cref{thm:k-level-bloomfilter}]
	\label{appendix:k-level-bloomfilter}
	Using communication on a $k$-dimensional
	grid, performing at most $\nmax$ operations (insertions, queries) per PE
	on a distributed single-shot Bloom Filter of size
	$m \ge en$ can be done in time
	$\mathcal{O}\left(k\left(\alpha
      p^{1/k}+\beta\nmax\log\frac{mp}{n} + \nmax\log{k}\right)\right)$
	in expectation and with probability $\ge 1 - 1/p^{\omega(1)}$
	assuming the total number of operations $n = \omega(k^2p^{1+1/k}\log{p})$ and additionally $m = \poly(n)$.
\end{theorem}
\begin{proof}(Outline)
	We focus on the part that is different from~\cite{sanders2013communication}.
	We route operations to their recipient PE determined by the hash value of the corresponding elements using a $k$-dimensional grid.
	This takes $k$ rounds of communication.
	In each of these rounds, there are $r = \sqrt[k]{p}$ messages per PE incurring a latency overhead of $\mathcal{O}(\alpha \sqrt[k]{p})$.
	In each round, we have the invariant that the received operations, i.e., hash values that are either queried or eventually inserted, are increasingly sorted.
	If a hash value occurs multiple times, it is forwarded only once.
	Due to the grid communication scheme, in each round $\recursionlevel$ with $0 \le \recursionlevel < k$, PEs exchange messages only with PEs located within the same row for dimension $\recursionlevel$ of the $k$-dimensional grid.
	These PEs constitute a group $g$.
    In the following, we will show that the probability that any PE in round $\recursionlevel$ receives more than $\nmax(1 + 1/k)^\recursionlevel$ hash values is in $1/p^{\omega(1)}$.

	We first consider the number of operations received by a PE $j$ in the first round and then show the claim by induction.
	Without loss of generality, we assume that each PE stores $\nmax$ operations.
	While we can assume that hash values associated with the (insert/query) operations are locally unique, we have the situation that they can occur multiple (at most $r$) times across the PEs within a PE-group $g$.
	By \cref{appendix:lemma:balls-into-bins-with-duplicates}, we know that each PE in $g$ receives more than $\nmax(1 + 1/k)$ operations with probability $\le 1/p^{\omega(1)}$.
	This holds since we assume $\nmax \ge n/p = \omega(k^2r\log{p})$ and therefore fulfill the condition on $\mbar$ in \cref{appendix:lemma:balls-into-bins-with-duplicates} with $p' = r$.
	As there are at most $\mathcal{O}(p)$ such groups, by the union bound, the probability that a PE in any group receives more than $\nmax(1+1/k)$ hash values is also $\le 1/p^{\omega(1)}$.
	Since duplicate hash values are only forwarded once, the assumption that each hash value occurs at most $r$ times within a PE group also holds for the following dimension.

	Let us assume that in round $0 < \recursionlevel < k-1$ a PE has received at most $\nmax (1+1/k)^\recursionlevel$ hash values with probability $\le 1/p^{\omega(1)}$ and each hash values occurs at most $r$ times
	in the PE group $g$ for round $\recursionlevel+1$.
	We apply \cref{appendix:lemma:balls-into-bins-with-duplicates} with $\mbar = \nmax (1+1/k)^\recursionlevel$ and $p' = r$ and therefore have $\mbar = \nmax (1+1/k)^\recursionlevel \ge \nmax \ge n/p = \omega(k^2r\log{p})$.
	Using the union bound argument as before, we find the event that any PE receives more than $\mbar(1+1/k) = \nmax (1+1/k)^\recursionlevel(1+1/k) = \nmax(1+1/k)^{\recursionlevel+1}$ occurs with probability in $1/p^{\omega(1)}$.
	Again as duplicate values are forwarded only once, we also fulfill the second inductive assumption.

	Thus, overall, we can conclude that any PE sends or receives at most $\nmax(1+k)^\recursionlevel \le \nmax(1+k)^k \le \nmax e$ hash values with probability $\ge 1 - 1/p^{\omega(1)}$ in round $\recursionlevel$.

	We now argue about the required communication volume in bits when a PE exchanges at most $\nmax e$ hash values.

	With Elias-Fano encoding \cite{Elias74,Fano71}, encoding $x$ out of $m$ one-bits can be achieved using at most $\encoding{m}{x} = x\left(\log\frac{m}{x}+2\right)$ bits.
    Therefore, a message containing $x$ operations on a Bloom filter of size $m$ can be encoded with at most $\encoding{m}{x}$ bits. 
	In each round, each PE sends (and analogously receives) $r$ messages containing $n_1, \dots, n_r$ operations with $\sum_{j=1}^r {n_j} = n'$ and we find $n' \le e\nmax$ with high probability.
	These messages can be encoded using
	$ \sum_{j=1}^r \encoding{\frac{m}{r}}{n_j}$ bits
	\footnote{Note that we can use $m/r$ as the universe size within the encoding, as each PE only receives hash values from an interval of size $m/r$.
  This could even be generalized to $m/r^{\recursionlevel+1}$ for round $\recursionlevel$.}.
	Due to the concavity of $\encoding{\frac{m}{r}}{\cdot}$, we obtain the following bound on the communication volume
	\begin{align*}
		\sum\limits_{j=1}^r \encoding{\frac{m}{r}}{n_j} & \le r \cdot \encoding{\frac{m}{r}}{\frac{n'}{r}} \le r \cdot \encoding{\frac{m}{r}}{\frac{e\nmax}{r}} \\
		                                                & = e\nmax  \left(\log{\frac{m}{e\nmax}} + 2\right)                                                     \\
		                                                & = \mathcal{O}\left(\nmax\left(\log{\frac{mp}{n}} + 1\right)\right)
	\end{align*}
	with high probability as $\encoding{\frac{m}{r}}{x}$ is monotonically increasing in its second component for $x \le 4m/(er)$ and we have $n/p \le \nmax \le n$ and $m \ge en$.
	The claim about the expected running time follows directly for $n$ being polynomial in $p$, for even larger $n$ one has to refine the estimation in \cref{appendix:lemma:balls-into-bins-with-duplicates}.

	The answers (yes or no) to the queries are sent on the reverse route using a single bit per request.
	Routing them to the PE that originally sent each request is possible by keeping appropriate tables mapping bit positions to sender PEs.
	These tables are stored locally and need no additional communication volume.

	The local work in each round is dominated by (integer-) sorting the received hash values.
    This can be achieved in $\mathcal{O}(\nmax\log{k})$ time with high probability \cite{DBLP:journals/tcs/KirkpatrickR84} as we have at most $\mathcal{O}(\nmax)$ values on a PE in any round with high probability and assume $m = \poly(n)$.
	The initial sorting of the at most $\nmax$ requires one additional sorting step.
    Thus, we end up with $\mathcal{O}(\nmax k \log{k})$ local work with high probability.
\end{proof}
\subsection{Duplicate Detection with Atomic Hypercube Quicksort}
In the following, we will describe how atomic hypercube quicksort \cite{axtmann2015practical, DBLP:phd/dnb/Axtmann21} can be used for distributed duplicate detection.
\begin{theorem}
	\label{appendix:thm:qsort_duplicate}
	For $n = \mathcal{O}(\poly(p))$ (atomic) hypercube quicksort \cite[Theorem 6.6]{DBLP:phd/dnb/Axtmann21} can be used to perform distributed duplicate detection with constant false positive probability $f^+$
	in time
	\[
		\mathcal{O}\parens*{\nmax\log{p} + \log^2{p}(\alpha + \beta(\nmax + \log{p})}
	\] in expectation and with probability $\ge 1-p^{c}$ for any constant $c > 0$
	where $\nmax$ is the maximum number of elements located on a single PE.
\end{theorem}
\begin{proof}
	As a first step, we map each element to an integer in \([0,nf^+)\) using a random uniform hash function \(h\).
	It is then sufficient to globally check whether there are duplicate hash values to achieve duplicate detection with false positive probability $f^{+}$ \cite{sanders2013communication}.

	After hashing each local element $e_j$, we construct hash value/index pairs $(h(e_j), j)$ which are then subsequently sorted using atomic hypercube quicksort \cite{DBLP:phd/dnb/Axtmann21}.
	Note that $j$ denotes the global index which can be computed using a distributed prefix sum over the number of elements on each PE.
	Since all indices are unique, the hash value/index pairs are also unique.
	During the sorting process, we keep track of the path a pair takes on each level by storing this information into a local table.
	Once the pairs are (lexicographically) sorted, a scan over the local elements and two additional message exchanges for the first and last local element suffices to identify duplicates.
	To communicate the result back, we replace the hash value within each hash value/index pair with a single bit indicating whether the element is unique or not.
	Then the pairs are sent back in $\mathcal{O}(\log{p})$ iterations using the stored routing information.

	Since $n$ is polynomial in $p$, we find $\log{nf^{+}} = \mathcal{O}(\log{p})$.
	Therefore, each hash value/index pair occupies only $\mathcal{O}(\log{p})$ bits and we achieve the above stated high probability bound.
	We need $\mathcal{O}(n^{c'})$ time per level in case all elements accumulate on a single PE during the course of the algorithm, which is highly unlikely as this event occurs with probability smaller than $p^{c}$ for an arbitrary large constant $c$.
	Since we assume $n$ to be polynomial in $p$, the above bound on the running time also holds in expectation.
\end{proof}

\section{Distinguishing Prefix Approximation via Doubling}
\label{appendix:section:prefix-doubling}

In the following, we prove the running time of our distinguishing prefix approximation algorithm using $k$-level Bloom filters and (atomic) hypercube quicksort for duplicate detection.

\begin{theorem}[\cref{thm:prefix-doubling}]
  \label{appendix:thm:prefix-doubling}
	For each string $s \in \mathcal{S}$ with $\distMath{s} \ge \log{p}/\log{\sigma}$ an approximation $\distApproxMath{s}$ with $\expectedValue{\distApproxMath{s}} = \mathcal{O}(\distMath{s})$ can be computed in time
    \[
      \mathcal{O}\left( \overbrace{\alpha k\sqrt[k]{p}\log{\dmax}}^{\mathrm{latency}} + \overbrace{\beta k \left(\frac{n}{p}\log{p}+ \frac{D}{p}\log{\sigma}\right)}^{\mathrm{communication \ volume}}  + \overbrace{k\frac{n}{p}\log{k}\log\log{\sigma} + \frac{D}{p}}^{\mathrm{local \ work}}\right)
	\]
	in expectation.
	We assume a balanced distribution of strings and their distinguishing prefixes, i.e., $\Theta(n/p)$ strings and $\Theta(D/p)$ per PE, and an overall number of strings $n = \mathcal{O}(\poly(p))$.
    Additionally, we assume $n/p = \omega(k^2\sqrt[k]{p}\log{p}\log\log{p})$ and $k \le \log{p}/(2\log\log{p})$.
\end{theorem}
\begin{proof}(Outline)
    We first briefly discuss the overall algorithmic idea to obtain the above stated running time and then discuss details.
    We start with an initial prefix length $l^{init} = \Theta(\log{p}/\log{\sigma})$ for the prefix doubling process using a $k-$level Bloom filter (see \cref{appendix:k-level-bloomfilter}) with a constant false positive probability $f^+$.
	By doubling the tested prefix length $l$ in each round, this results in $\mathcal{O}(\log{\dmax})$ rounds of duplicate detection in expectation and therefore an expected overall latency in $\mathcal{O}(\alpha\log{\dmax}k\sqrt[k]{p})$.
	We therefore no longer have to deal with the latency term when arguing about the expected running time.
    A limitation of $k$-level Bloom filters is that they require a certain number of operations $n^{thr}$ so that a running time in $\mathcal{O}(k(\beta\log{p} + \log{k}))$ per operation can be achieved (see \cref{appendix:k-level-bloomfilter}).
    However, during the course of the prefix doubling process, more and more strings drop out as an approximation of their distinguishing prefix has been found.
    We mitigate this problem by switching to (atomic) hypercube quicksort once the number of strings becomes too small.

    $\mathbf{l \le l^{thr}}$. We now start with a discussion of the running time until we have reached a tested prefix $l = l^{thr}$ with $l^{thr} = \log^3(p)$.
	Let $\mathcal{X}_{l}$ denote the maximum number of strings participating in the duplicate detection round with prefix length $l$ on any PE.
    As stated in \cref{appendix:k-level-bloomfilter} ($k$-level Bloom filter), with probability $1-1/p^{\omega(1)}$ we have a running time for each round of duplicate detection in $\mathcal{O}(\mathcal{X}_l \cdot k(\beta\log{p} + \log{k}))$.
    As an upper bound for the running time in the unlikely case where all operations are routed to the same PE occurring with probability in $1/p^{\omega(1)}$, we find $\mathcal{O}(X_l\cdot pk(\beta \log{p}+ \log{k}))$.
    Therefore, by the linearity of expectation, the expected running time is in $\mathcal{O}(\expectedValue{\mathcal{X}_l} k(\beta\log{p} + \log{k}))$ in each round.
    In total, we therefore have $\sum_{i=1}^{c\log\log{p}} \mathcal{O}(\expectedValue{\mathcal{X}_{2^i} l_{\textrm{init}}}k(\log{p}\beta + \log{k}))$ as an upper bound on the expected running time for rounds $l^{\textrm{init}} < l \le l^{\textrm{thr}}$ with a constant $c > 0$ as it takes $\mathcal{O}(\log\log{p})$ rounds to reach $l = l^{thr}$.

    In the first round, we have an expected running time in $\mathcal{O}(n/pk(\beta\log{p} + \log{k}))$ as all strings participate.
    It remains to bound the size of the expected value of $\mathbb{E}(X_{l})$ for the remaining rounds which can be done using the average distinguishing prefix length $\davg = D/n$.
    We find $\davg = (D/p)/(n/p) = D/n = \Omega(\log{n}/\log{\sigma}) = \Omega(\log{p}/ \log{\sigma})$ on each PE as we assume all strings to be unique.
    A general observation is that there are at most $n/(p2^j)$ strings with distinguishing prefix length greater than $\davg 2^j$ for any $j \ge 0$.
    Assuming a constant false positive probability, for $l \ge \davg$, we have $\sum_{\davg \le l} \expectedValue{\mathcal{X}_l} = \mathcal{O}(n/p)$.
    Therefore, the expected running time for those rounds is dominated by the first one.
    This also holds when the number of participating strings drops below holds $n^{thr} = \omega(k^2p^{1/k}\log{p})$ during the first $\mathcal{O}(\log\log{p})$ rounds, since we have $\mathcal{O}(n^{thr}k(\beta\log{p} + \log{k}))$ as an upper bound on the expected running time in these cases.
    Since we assume $n/p = \omega(k^2p^{1/k}\log{p}\log\log{p}) = \Omega(n^{thr}\log\log{p})$, the accumulated running time for those $\mathcal{O}(\log\log{p})$ rounds is bounded by the running time of the first round

    We now differentiate three cases:
    \begin{enumerate}
      \item $\davg \le l^{\textrm{init}}.$ In this case, the above argument directly applies and the accumulated running time for the first $\mathcal{O}(\log\log{p})$ rounds is dominated by the running time of the first round.

      \item 
      \label{appendix:proof:prefix-doubling:case2} 
        $l^{\textrm{init}} < \davg \le \log{p}.$ 
        \noindent
        In this case, there are $\mathcal{O}(\log(\davg \log{\sigma}/ \log{p})) = \mathcal{O}(\log\log{\sigma})$ rounds until we reach $l = \davg$.
        Assuming that all $n/p$ strings on a PE participate in these rounds  --  note that this is an upper bound on $\mathcal{X}_l$  --  causes an expected communication volume in 
        \[
          \mathcal{O}\left(\log\left(\frac{\davg \log{\sigma}}{\log{p}}\right)\frac{n}{p}k\log{p}\right) = \mathcal{O}\left(\frac{\davg \log{\sigma}}{\log{p}} \frac{n}{p} k\log{p}\right) = \mathcal{O}\left(k \frac{D}{p} \log{\sigma}\right).
        \]
        For expected local work, we have to charge $\mathcal{O}(\log\log{\sigma} \cdot n/pk\log{k})$.
        The subsequent rounds with $l > \davg$ are then dominated by the running time of the very first round.

        \item
        \label{appendix:proof:prefix-doubling:case3} 
        $\log{p} < \davg.$ For the rounds until we reach $l = \log{p}$, we have the same running time as in case \ref{appendix:proof:prefix-doubling:case2}.
        For the remaining $\mathcal{O}(\log(\davg/\log{p})$ rounds until we reach $l = \min(\davg, l^{\textrm{thr}})$, we may also assume that all strings further participate as an upper bound on $\mathcal{X}_l$.
        We then have a running time in
\[
  \mathcal{O}\left(\log\left(\frac{\davg}{\log{p}}\right)\frac{n}{p}k(\log{p}\beta + \log{k})\right) = \mathcal{O}\left(\frac{\davg}{\log{p}} \frac{n}{p}k(\log{p}\beta + \log{k})\right) = \mathcal{O}\left(\frac{D}{p}(k\beta + 1)\right).
        \]
        Note that $k\log{k} = \mathcal{O}(\log{p})$ as we assume $k \le \log{p}/\log\log{p}$.
        Potential further rounds with  $\davg < l \le l^{thr}$ are dominated by the running time of the first round.
    \end{enumerate}
    Therefore, we find a running time for the first $\mathcal{O}(\log\log{p})$ iterations as stated in the theorem above.

    $\mathbf{l > l^{thr}}$.
    Until now we have analyzed the expected running time of the algorithm until a tested prefix length $l = l^{thr} =  \Omega(\log^3{p})$ is reached.
    From now on we switch to duplicate detection with hypercube quicksort once $X_l$ drops below $n^{thr}$ (see \cref{appendix:thm:qsort_duplicate}).
    Since the running time of duplicate detection via $k$-level Bloom filters is dominated by that of duplicate detection via hypercube quicksort (apart from the latency term which has already been dealt with) we can assume without loss of generality that only hypercube quicksort is used from now on.
    As stated in \cref{appendix:thm:qsort_duplicate}, with probability $1-1/p^{c_1}$ for an arbitrary large constant $c_1$ we have a running time for each round of duplicate detection in $\mathcal{O}(\mathcal{X}_l\log^2{p}(\beta + 1) + \mathbbm{1}_{\mathcal{X}_l > 0}\beta\log^3{p})$ with $\mathbbm{1}_{\mathcal{X}_l > 0}$ being the indicator variable whether there are still strings in the duplicate detection process.
    As an upper bound for the running time in the unlikely case occurring with probability less than $1/p^{c_1}$, we find $\mathcal{O}(n^{c_2})$.
    This is in $\mathcal{O}(p^{c_3})$ as $n$ is polynomial in $p$ with $c_2,c_3 > 0$ being fixed constants.
    By the linearity of expectation and by choosing $c_1 > c_3$, we find the expected running time to be in $\mathcal{O}(\expectedValue{\mathcal{X}_l}\log^2{p}(\beta + 1) + \mathcal{O}(\expectedValue{\mathbbm{1}_{\mathcal{X}_l > 0}}\log^3{p}\beta)$ in each round.

    We begin by finding a bound for the latter term.
    If the longest distinguishing prefix $\dmax \le \log^3{p}$, we know that the expected number of rounds with $l > l^{thr}$ is constant.
    Thus, $\mathcal{O}(\expectedValue{\mathbbm{1}_{\mathcal{X}_l > 0}}\log^3{p}\beta) = \mathcal{O}(\log^3{p}\beta) = \mathcal{O}(n/p\log{p}k\beta)$ as $n/p = \omega(\log{p}\sqrt[k]{p})$ and $k \le \log{p}/\log\log{p}$.
    If the longest distinguishing prefix $\dmax > \log^3{p}$, we have $\dmax / \log^3{p} > 1$.
    Hence, we have $\log{\dmax / \log^3{p}} = \mathcal{O}(\dmax / \log^3{p})$ remaining rounds in expectation and the accumulated costs of $\mathcal{O}(\log^3{p}\beta)$ in each of these rounds are in $\mathcal{O}(\log^3{p} \beta \dmax / \log^3{p}) = \mathcal{O}(D/p)$, as $\dmax \le \mathcal{O}(D/p)$.

    Now we discuss the $\mathcal{O}(\expectedValue{\mathcal{X}_l}\log^2{p}(\beta + 1))$ part of the running time.
    We differentiate two cases:
    \begin{enumerate}
      \item $\mathbf{\davg < \log^2{p}}.$ The number of strings on any PE with distinguishing prefix length $\log^3{p} \ge \davg\log{p}$ or longer is smaller than $n/(p\log{p})$ as discussed above
      Therefore, $\sum_{l=l^{thr}}^\infty \expectedValue{\mathcal{X}_l} = \mathcal{O}(n/(p\log{p}))$.
      Hence the accumulated costs for duplicate detection with hypercube quicksort are in $\mathcal{O}(n/(p\log{p})\log^2{p}(\beta + 1) = \mathcal{O}(n/p \log{p}(\beta + 1)$ and thus are dominated by the running time of the very first iteration of the prefix doubling process.
      \item $\mathbf{\davg \ge \log^2{p}}.$ We have $D/p = \Omega(n/p\log^2{p})$. For the $\mathcal{O}(\log(\davg / \log^2{p}))$ rounds until $l$ reaches $\davg$, we can therefore simply assume that $n/p$ strings participate in the duplicate detection causing expected costs in
\[
  \mathcal{O}\left(\log\left(\frac{\davg}{\log^2{p}}\right)\frac{n}{p}\log^2{p}(\beta + 1)\right) = \mathcal{O}\left(\frac{\davg}{\log^2{p}} \frac{n}{p}\log^2{p}(\beta + 1)\right) = \mathcal{O}\left(\frac{D}{p}(\beta + 1)\right).
        \]
        The costs of the remaining rounds when $l \ge \davg$ are again in $\mathcal{O}(n/p\log^2{p}(\beta +1)) = \mathcal{O}(D/p(\beta + 1))$.
        In fact, if we recall that in these rounds we actually use the $k$-level Bloom filter until $n^{thr}$ is reached and switch to hypercube quicksort once we fall below this threshold,
        we even obtain a running time in $o(D/p(\beta + 1)$.
  \end{enumerate}

  Therefore, in every case the summed expected running time is in the bound stated in the theorem above.
\end{proof}

\section{More Experimental Results}
\label{appendix:sec:additional_results}
In this section, we give additional experimental results that highlight some of the key properties of our new distributed multi-level string sorting algorithms.
These results break the average running times down and give time requirements for different phases of the algorithms.

First, in \cref{fig:appendix:np_ratio}, we give the average sorting times per sorting phase for the weak scaling experiments on synthetic data with different \(n/p\) rations, i.e., the results presented in \cref{fig:evaluation:np_ratio}.
Here, we can see that our multi-level approach significantly reduces the time spent in the partitioning phase (both \ms\ and \pdms) and the approximation of the distinguishing prefix, i.e., the Bloom filter (\pdms\ only).
For the total running times, we also give more values for slower algorithms, which is feasible due to the size of the plot.
This further highlights the scalability of our new algorithms.
When not only considering \ms\ but also \pdms, we can achieve a speedup of up to \(7\).

Next, in \cref{fig:appendix:dn_ratio}, we depict the communication volume and the average sorting times per sorting phase for the weak scaling experiments on synthetic data with different \(D/N\) rations, i.e., the results presented in \cref{fig:evaluation:dn_ratio}.
Communication volume is measured on each PE during execution using the size of buffers passed to
MPI routines and summed afterward to arrive at the final value.
Different rules are applied depending on the used routine:\fullstop{}
For example, calls to \texttt{MPI\_Alltoall} and \texttt{MPI\_Alltoallv}, as well as reduce and
scan operations count the size of send buffers on all PEs.
Calls to \texttt{MPI\_Bcast} only count the send buffer on the root PE and multiply its size by
$p$.
The result is only an approximation of the actual communication volume for several.
Most importantly, the measurements are idealized and do not necessarily correspond to actual
communication performed by MPI -- especially for broadcast and reduction operations.
Communication of collective exchange operations may also be overestimated if send buffers include
data that is already on the correct PE and does not require sending.

Finally, in \cref{fig:appendix:strong_scaling}, we give the average running times per sorting phase for our strong scaling experiment using real-world inputs, which we depict in \cref{fig:evaluation:strong_scaling}.
The results highlight the difficulty of the real-world inputs as many benefits that are visible for synthetic data are less prominent.

\iffigures
	\begin{figure*}
		\centering
		\pgfplotsset{
    np ratio legend/.style={
            legend columns=8,
            legend style={at={(0.5, -0.3)}, anchor=north},
            legend entries={$\ms_1$\\$\pdms_1$\\$\ms_2$\\$\pdms_2$\\$\ms_3$\\$\pdms_3$\\
                    $\RQuick$\\$\LcpRQuick$\\},
        },
}

\begin{tikzpicture}[baseline]
    \begin{groupplot}[
            footnotesize,
            group style={
                    group size=3 by 1,
                    x descriptions at=edge bottom,
                    ylabels at=edge left,
                    yticklabels at=all,
                    horizontal sep=2em,
                },
            height=.22\textheight, width=.25\textwidth,
            scale only axis, enlargelimits=0.025,
            xtick=data, xmode=log, log basis x=2, xmajorgrids,
            xlabel={$\text{nodes}\times48/p$}, ylabel={$\text{wall time}/\unit{\second}$}, 
            xticklabel=\pgfmathparse{2^\tick}\pgfmathprintnumber{\pgfmathresult},
            x tick label style={rotate=90,/pgf/number format/.cd,fixed,precision=0},
            ylabel={$\text{wall time}/\unit{\second}$},
            label style={font=\footnotesize},
            title style={font=\footnotesize},
            cycle list name=colorListUpToThree,
        ]

        \nextgroupplot[title={$n/p = 10^4$}, ymin=0, ymax=1.6, ytick distance=0.1,
            y tick label style={/pgf/number format/.cd,fixed,fixed zerofill=true,precision=1}]

        \addplot coordinates { (4,0.0505905) (8,0.0662462) (16,0.102683) (32,0.216874) (64,0.449523) (128,1.16063) (256,1.17907) (512,1.95185) (1024,3.79291) };
        \addplot coordinates { (4,0.0418138) (8,0.0557095) (16,0.102141) (32,0.201617) (64,0.433726) (128,1.56009) (256,2.97662) (512,7.05657) (1024,18.6557) };
        \addplot coordinates { (4,0.0560272) (8,0.0587485) (16,0.068184) (32,0.0689761) (64,0.0794039) (128,0.129826) (256,0.174955) (512,0.379632) (1024,0.860719) };
        \addplot coordinates { (4,0.0373176) (8,0.0376604) (16,0.0470978) (32,0.0493042) (64,0.0603443) (128,0.0826471) (256,0.148347) (512,0.298922) (1024,0.66346) };
        \addplot coordinates { (4,0.0897979) (8,0.0910808) (16,0.0919553) (32,0.0936831) (64,0.106478) (128,0.131268) (256,0.163788) (512,0.266722) (1024,0.531366) };
        \addplot coordinates { (4,0.0474277) (8,0.0466564) (16,0.049892) (32,0.0539536) (64,0.0600533) (128,0.0780923) (256,0.109479) (512,0.19862) (1024,0.409279) };

        \addplot coordinates { (4,0.254327) (8,0.291817) (16,0.331052) (32,0.375952) (64,0.444785) (128,0.502336) (256,0.591295) (512,0.662481) (1024,0.86462) };
        \addplot coordinates { (4,0.226327) (8,0.26956) (16,0.300978) (32,0.342181) (64,0.407616) (128,0.459926) (256,0.88805) (512,0.989355) (1024,0.813704) };

        \nextgroupplot[
            title={$n/p = 10^5$},
            ymin=0, ymax=7.2, ytick distance=0.4,
            y tick label style={/pgf/number format/.cd,fixed,fixed zerofill=true,precision=1},
            np ratio legend
        ]

        \addplot coordinates { (4,0.644801) (8,0.677757) (16,0.813842) (32,0.832994) (64,1.0635) (128,1.75911) (256,2.88921) (512,5.53908) (1024,12.8824) };
        \addplot coordinates { (4,0.59318) (8,0.606499) (16,0.647423) (32,0.753289) (64,0.98564) (128,1.64705) (256,3.72144) (512,11.1687) (1024,24.9355) };
        \addplot coordinates { (4,0.831777) (8,0.828691) (16,0.839111) (32,0.921098) (64,0.955699) (128,1.01874) (256,1.05302) (512,1.22557) (1024,1.76392) };
        \addplot coordinates { (4,0.633879) (8,0.631919) (16,0.646015) (32,0.662322) (64,0.673849) (128,0.715472) (256,0.788839) (512,0.945914) (1024,1.30923) };
        \addplot coordinates { (4,1.22442) (8,1.20322) (16,1.13623) (32,1.14768) (64,1.19633) (128,1.27094) (256,1.33497) (512,1.50795) (1024,1.87613) };
        \addplot coordinates { (4,0.781282) (8,0.759182) (16,0.730639) (32,0.725067) (64,0.740623) (128,0.761897) (256,0.80277) (512,0.901721) (1024,1.07139) };

        \addplot coordinates { (4,2.88197) (8,3.2947) (16,3.70603) (32,4.14237) (64,4.87946) (128,5.411) (256,6.28294) (512,6.90788) (1024,9.06809) };
        \addplot coordinates { (4,2.64504) (8,3.03981) (16,3.40672) (32,3.78964) (64,4.48999) (128,4.99303) (256,5.79771) (512,6.38561) (1024,8.51642) };

        \nextgroupplot[title={$n/p = 10^6$}, ymin=0, ymax=16, ytick distance=2]

        \addplot coordinates { (4,5.20592) (8,5.43991) (16,5.62522) (32,5.78779) (64,6.57597) (128,8.23696) (256,9.33346) (512,12.3119) (1024,24.1115) };
        \addplot coordinates { (4,4.62864) (8,4.77387) (16,4.83676) (32,5.12225) (64,5.35129) (128,6.57427) (256,8.54862) (512,15.2255) (1024,31.2713) };
        \addplot coordinates { (4,7.17933) (8,7.1421) (16,7.21363) (32,7.37185) (64,7.87024) (128,8.39283) (256,8.3063) (512,8.41634) (1024,11.6387) };
        \addplot coordinates { (4,5.23343) (8,5.23599) (16,5.29987) (32,5.38873) (64,5.41789) (128,5.47008) (256,5.61729) (512,5.84161) (1024,6.97526) };
        \addplot coordinates { (4,11.3501) (8,11.0116) (16,10.3207) (32,10.3167) (64,10.5077) (128,11.4675) (256,11.4457) (512,11.8219) (1024,14.4682) };
        \addplot coordinates { (4,6.65732) (8,6.44196) (16,6.14305) (32,6.1952) (64,6.11161) (128,6.17939) (256,6.32011) (512,6.53512) (1024,7.37365) };
    \end{groupplot}
\end{tikzpicture}
		\input{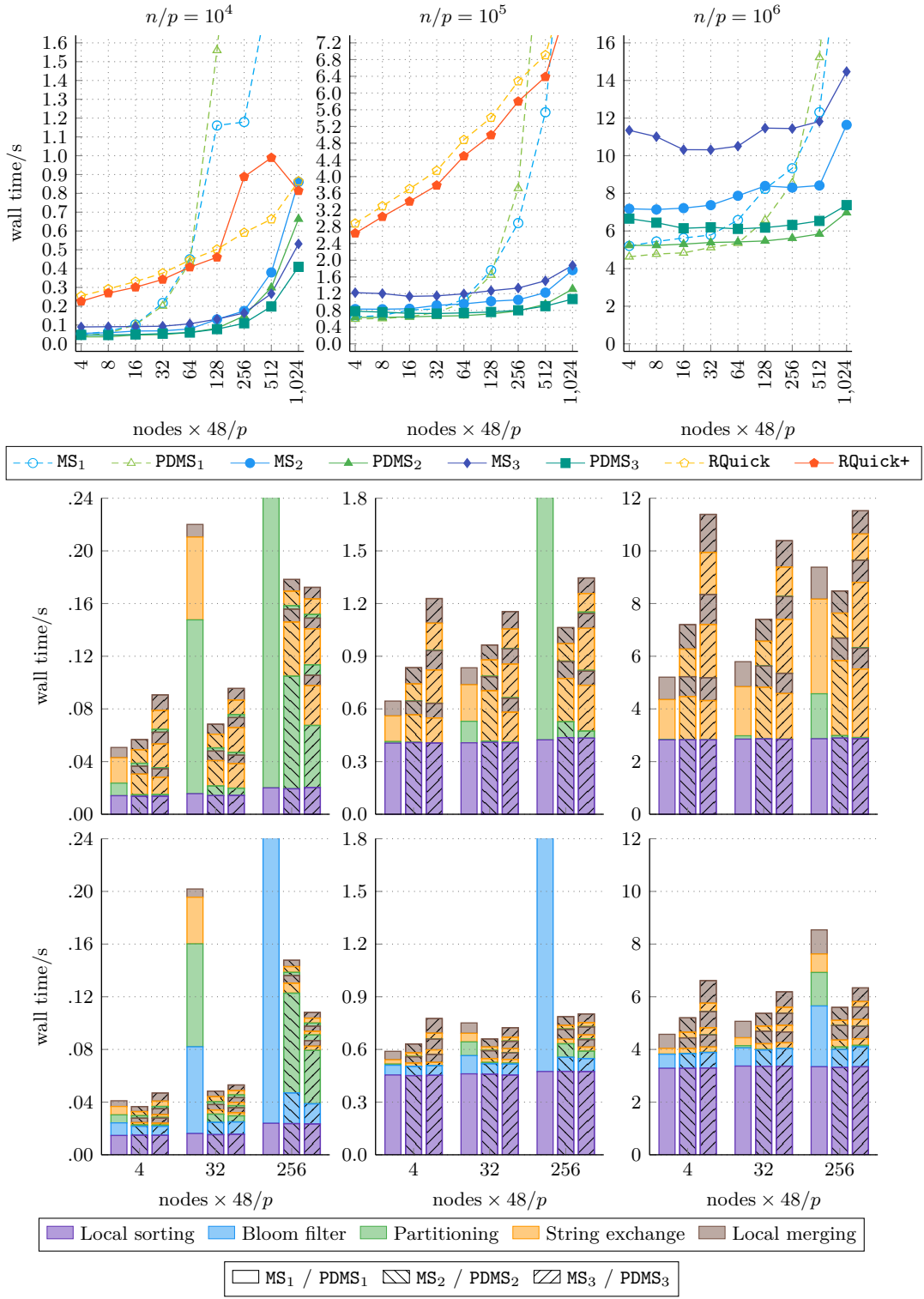}
		\caption{
			 Average sorting times (top row, same as \cref{fig:evaluation:np_ratio} but with wider $y$-range) and running times per sorting phase for $\ms_k$ (middle
			row) and $\pdms_k$ (bottom row) for a weak scaling experiment using \textsc{DNData}
			with $\ell = 500$ and $D/N = 0.5$. Phases are in execution order starting from the bottom.}
		\label{fig:appendix:np_ratio}
	\end{figure*}

	\begin{figure*}[h]
		\centering
		\pgfplotsset{
	dn ratio time style/.style={
			ymin=0, ymax=6,
			ytick distance=1
		},
	dn ratio comm style/.style={
			ymin=0, ymax=3.0,
			ytick distance=0.5,
			y tick label style={
					/pgf/number format/.cd,
					fixed,fixed zerofill,precision=1
				},
		}
}

\begin{tikzpicture}[baseline]
	\begin{groupplot}[
			footnotesize,
			group style={
					group size=5 by 1,
					x descriptions at=edge bottom,
					y descriptions at=edge left,
					horizontal sep=.5em,
					vertical sep=1ex,
				},
			xmajorgrids, scale only axis,
			enlarge x limits=0.05,
			enlarge y limits=0.03,
			height=.215\textheight, width=.1685\textwidth,
			xmin=4, xmax=128,
			xtick={4, 8, 16, 32, 64, 128}, xmode=log, log basis x=2,
			xticklabel=\pgfmathparse{2^\tick}\pgfmathprintnumber{\pgfmathresult},
			x tick label style={rotate=90,/pgf/number format/.cd,fixed,precision=0},
			label style={font=\footnotesize}, title style={font=\footnotesize},
      cycle list name=colorListUpToTwoWithoutPSV,
		]

		\nextgroupplot[title={$D/N = 0.0$},ylabel={$\text{bytes sent per string}/\unit{\kilo\byte}$}, dn ratio comm style]

  \addplot coordinates { (4,0.513891) (8,0.524763) (16,0.548247) (32,0.598905) (64,0.707945) (128,0.940802) };
  \addplot coordinates { (4,1.01197) (8,1.01245) (16,1.0132) (32,1.01451) (64,1.01709) (128,1.02042) };
  \addplot coordinates { (4,0.0377232) (8,0.0385521) (16,0.0400649) (32,0.0430121) (64,0.049075) (128,0.0608668) };
  \addplot coordinates { (4,0.0742054) (8,0.0744699) (16,0.0747601) (32,0.0750752) (64,0.0755011) (128,0.0742184) };

  \addplot coordinates { (4,1.92061) (8,2.17113) (16,2.4217) (32,2.67235) (64,2.92309) (128,3.17409) };
  \addplot coordinates { (4,1.94863) (8,2.20318) (16,2.45758) (32,2.71214) (64,2.96662) (128,3.22112) };

		\nextgroupplot[title={$D/N = 0.25$},dn ratio comm style]

  \addplot coordinates { (4,0.395121) (8,0.40622) (16,0.430162) (32,0.481733) (64,0.592605) (128,0.830057) };
  \addplot coordinates { (4,0.774033) (8,0.774516) (16,0.775277) (32,0.77661) (64,0.779224) (128,0.784627) };
  \addplot coordinates { (4,0.0410151) (8,0.0445929) (16,0.0520625) (32,0.0678439) (64,0.101417) (128,0.172855) };
  \addplot coordinates { (4,0.0767177) (8,0.0770405) (16,0.0774558) (32,0.0780398) (64,0.0790405) (128,0.0809801) };

  \addplot coordinates { (4,1.92063) (8,2.17108) (16,2.42176) (32,2.67233) (64,2.92307) (128,3.17409) };
  \addplot coordinates { (4,1.94867) (8,2.20303) (16,2.45755) (32,2.71205) (64,2.9666) (128,3.22115) };

  \nextgroupplot[title={$D/N = 0.5$},dn ratio comm style,
  			title={$D/N = 0.5$},
			dn ratio time style,
			legend columns=7,
			legend style={at={(0.5,-0.15)},anchor=north,font=\small},
			legend entries={$\ms_1$\\$\ms_2$\\$\pdms_1$\\$\pdms_2$\\\RQuick\\\LcpRQuick\\}]

  \addplot coordinates { (4,0.27036) (8,0.281699) (16,0.306121) (32,0.358653) (64,0.471436) (128,0.712731) };
  \addplot coordinates { (4,0.524096) (8,0.524586) (16,0.525357) (32,0.52671) (64,0.529363) (128,0.534847) };
  \addplot coordinates { (4,0.0464064) (8,0.0528613) (16,0.0665777) (32,0.0958342) (64,0.158321) (128,0.291517) };
  \addplot coordinates { (4,0.0832495) (8,0.0836331) (16,0.0841792) (32,0.0850457) (64,0.0866499) (128,0.0898765) };

  \addplot coordinates { (4,1.92049) (8,2.17106) (16,2.42161) (32,2.67238) (64,2.92315) (128,3.1741) };
  \addplot coordinates { (4,1.94859) (8,2.2031) (16,2.4576) (32,2.71203) (64,2.96659) (128,3.22111) };

		\nextgroupplot[title={$D/N = 0.75$},dn ratio comm style]

  \addplot coordinates { (4,0.145599) (8,0.15718) (16,0.182079) (32,0.235575) (64,0.35028) (128,0.595418) };
  \addplot coordinates { (4,0.274161) (8,0.274655) (16,0.275437) (32,0.276811) (64,0.279504) (128,0.285066) };
  \addplot coordinates { (4,0.161731) (8,0.17344) (16,0.198603) (32,0.252616) (64,0.368372) (128,0.615592) };
  \addplot coordinates { (4,0.306199) (8,0.306693) (16,0.307481) (32,0.308865) (64,0.311583) (128,0.317186) };

  \addplot coordinates { (4,1.92063) (8,2.17102) (16,2.42174) (32,2.67236) (64,2.9231) (128,3.17411) };
  \addplot coordinates { (4,1.94852) (8,2.20304) (16,2.4575) (32,2.71211) (64,2.96658) (128,3.22109) };

		\nextgroupplot[title={$D/N = 1.0$},dn ratio comm style]

  \addplot coordinates { (4,0.0208391) (8,0.0326584) (16,0.0580393) (32,0.112495) (64,0.229123) (128,0.478091) };
  \addplot coordinates { (4,0.0242271) (8,0.0247255) (16,0.0255168) (32,0.0269103) (64,0.0296441) (128,0.035289) };
  \addplot coordinates { (4,0.036972) (8,0.0489224) (16,0.0745623) (32,0.129538) (64,0.247214) (128,0.498271) };
  \addplot coordinates { (4,0.0562631) (8,0.0567647) (16,0.0575614) (32,0.0589674) (64,0.0617213) (128,0.0674084) };

  \addplot coordinates { (4,1.92063) (8,2.17117) (16,2.42169) (32,2.67233) (64,2.92313) (128,3.17409) };
  \addplot coordinates { (4,1.94847) (8,2.20311) (16,2.45759) (32,2.71207) (64,2.96661) (128,3.22116) };

	\end{groupplot}
\end{tikzpicture}
		\input{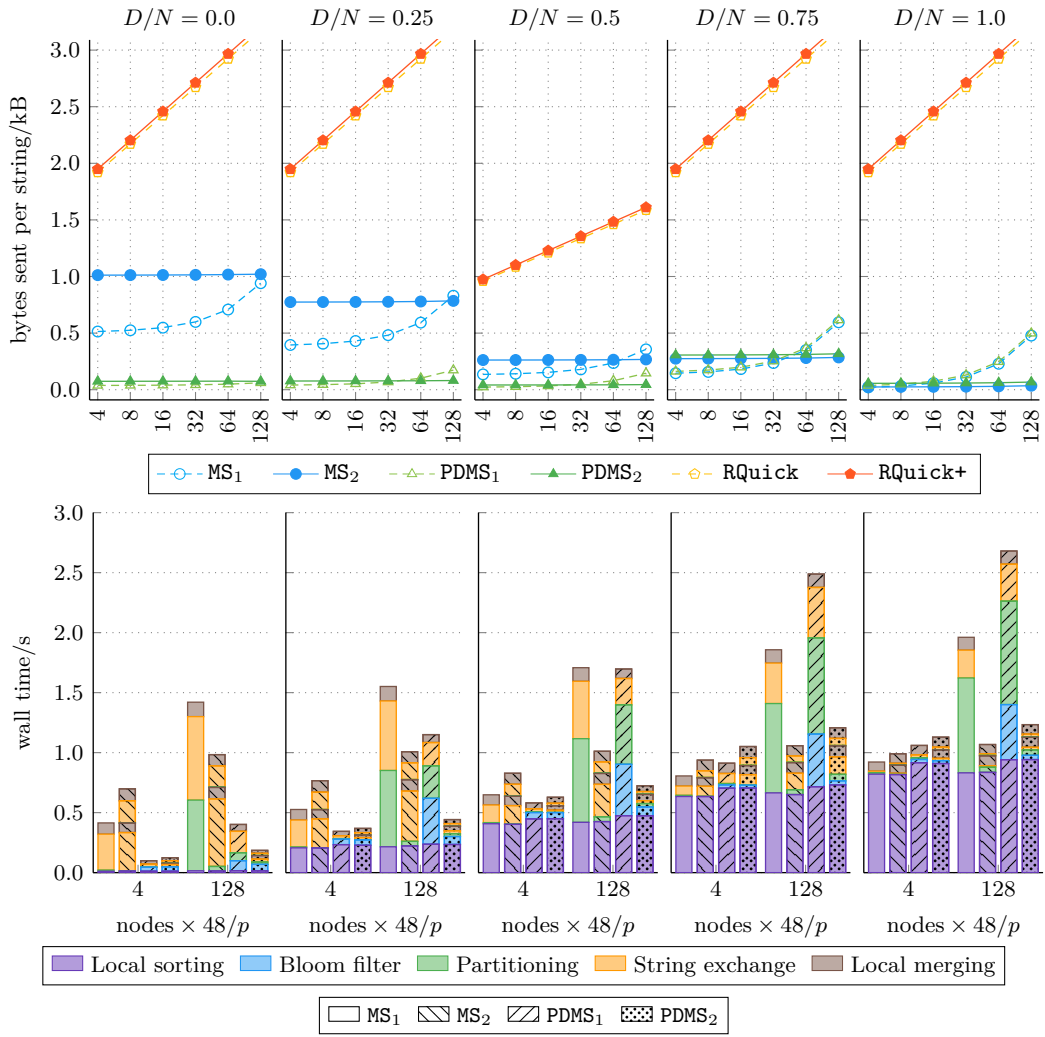}
		\caption[%
			Overall sorting times, bytes sent per string, and runtimes per sorting phase for the weak scaling
			experiment with variable $D/N$ ratio.
		]{%
			Approx. bytes sent per string (top row) and running time for the sorting phases (bottom row) for the weak scaling experiment using \textsc{DNData} inputs with $\ell=500$ and $n/p=10^5$ depicted in \cref{fig:evaluation:dn_ratio}. Sorting phases are in order of execution starting from the bottom.}
		\label{fig:appendix:dn_ratio}
	\end{figure*}

	\begin{figure*}
		\centering
		\input{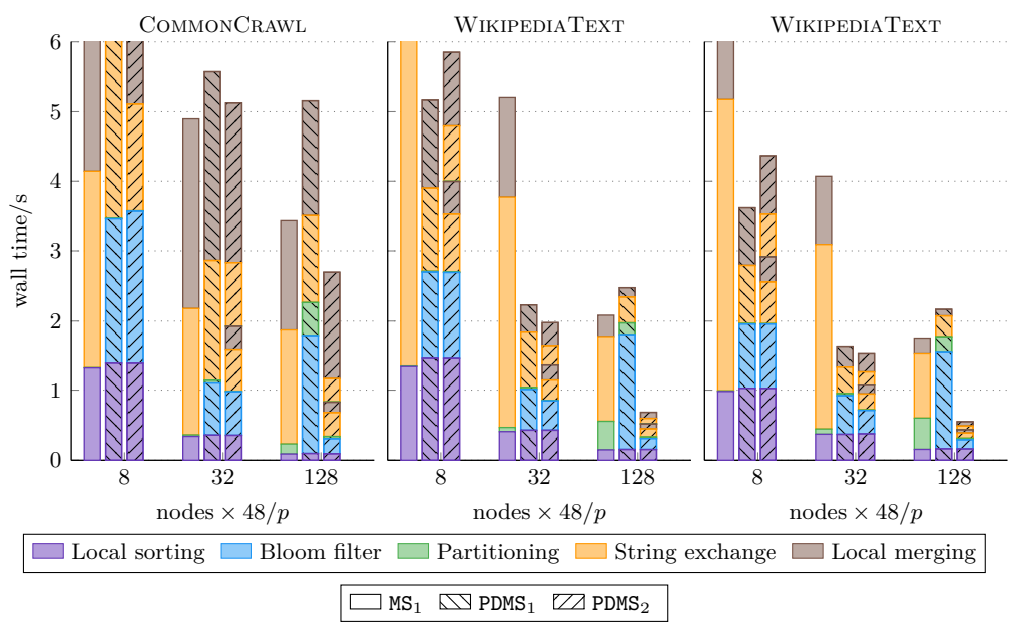}
		\caption{Average running times per sorting phase for the strong scaling
			experiment using real-world inputs depicted in \cref{fig:evaluation:strong_scaling}.
			Sorting phases are in order of execution starting from the bottom}
		\label{fig:appendix:strong_scaling}
	\end{figure*}
\fi

\end{document}